\documentclass[11pt,english]{article}

\pdfoutput=1 

\usepackage{scalefnt} 

\usepackage[utf8]{inputenc}
\usepackage[T1]{fontenc}
\usepackage[english]{babel}
\usepackage[dvipsnames]{xcolor} 
\usepackage{amsmath} 
\usepackage{amssymb} 
\usepackage{amsthm}  
\usepackage{hyperref}  
\usepackage{tikz}      
\usepackage{tcolorbox} 
\usepackage{dcolumn}   
\usepackage{graphicx, wrapfig, subcaption, setspace, booktabs, float, epsfig} 
\usepackage{multirow}  
\usepackage[all]{xy}

\usepackage{enumitem}
\usepackage{tikz}

\theoremstyle{plain}
\newtheorem{theorem}{Theorem}[section]

\newtheorem{conjecture}{Conjecture}[section]

\newtheorem{conclusion}{Conclusion}

\newtheorem{lemma}{Lemma}[section]

\newtheorem{corollary}{Corollary}
\newtheorem{example}{Example}

\theoremstyle{definition}
\newtheorem{remark}{Remark}

\hypersetup{
colorlinks=true,
linkcolor=blue,
citecolor=blue,
filecolor=blue,
urlcolor=black,
pdfauthor={Luiz Felipe Andrade Campos},
pdftitle={Reducible Yang-Mills Theories}}

\begin{document}
\title{On Extensions of Yang-Mills-Type Theories, Their Spaces and Their Categories}

\author{Yuri Ximenes Martins\footnote{yurixm@ufmg.br (corresponding author)}, Luiz Felipe Andrade Campos\footnote{luizfelipeac@ufmg.br},  Rodney Josu\'e Biezuner\footnote{rodneyjb@ufmg.br}\\ \\ \textit{Departamento de Matem\'atica, ICEx, Universidade Federal de Minas Gerais,}  \\  \textit{Av. Ant\^onio Carlos 6627, Pampulha, CP 702, CEP 31270-901, Belo Horizonte, MG, Brazil}}


\maketitle
\begin{abstract}
In this paper we consider the classification problem of extensions of Yang-Mills-type (YMT) theories. For us, a YMT theory differs from the classical Yang-Mills theories by allowing an arbitrary pairing on the curvature. The space of YMT theories with a prescribed gauge group $G$ and instanton sector $P$ is classified, an upper bound to its rank is given and it is compared with the space of Yang-Mills theories. We present extensions of YMT theories as a simple and unified approach to many different notions of deformations and addition of correction terms previously discussed in the literature. A relation between these extensions and emergence phenomena in the sense of \cite{emergence_yuri} is presented. We consider the space of all extensions of a fixed YMT theory $S^G$ and we prove that for every additive group action of $\mathbb{G}$ in $\mathbb{R}$ and every commutative and unital ring $R$, this space has an induced structure of $R[\mathbb{G}]$-module bundle. We conjecture that this bundle can be continuously embedded into a trivial bundle. Morphisms between extensions of a fixed YMT theory are defined in such a way that they define a category of extensions. It is proved that this category is a reflective subcategory of a slice category, reflecting some properties of its limits and colimits.
\end{abstract}
\section{Introduction}
\quad\;\,Although the Standard Model of Elementary Particle Physics (SM) is at present probably the most accurate and tested physical theory \cite{review_particle_physics}, there are indications that it must be viewed as an effective theory of a more fundamental theory \cite{beyond_SM_1,beyond_SM_2}. A lot of different approaches have been suggested along the last decades, e.g., string theory \cite{string} and other approaches involving higher-dimensional spacetime \cite{higher_dimension_1,higher_dimension_2}, loop quantum theory and other attempts to quantum gravity \cite{QG_1,QG_2}, symmetry breaking extensions such as the Standard-Model Extension \cite{SM_extension_1,SM_extension_2,SM_extension_3}, symmetry enlargement models such as the Minimal Supersymmetric Standard Model and other supersymmetric extensions with higher supersymmetry \cite{next_minimal_SM,next_minimal_SM_2,next_minimal_SM_3}, and so on.  This \textit{zoo} of extensions naturally leads one to consider the \textit{classification problem} of all already existing and all the possible extensions that could eventually be discovered in the future, avoiding spending time studying informal or superficial theories (which is clearly not the case of the approaches cited above). In this paper we propose the beginning of a program attempting to formalize and give a strategy of attack for such a classification problem.

Notice that since the SM can be built by coupling spinorial field theories with Yang-Mills (YM) theories and with its classical Higgs fields, one can see that the basic strategy is to begin by classifying the possible extensions in each component piece. In this paper and in the next two \cite{complete_extension,gauge_breaking_YMT_extensions} we will focus on the YM part. Actually, in this work we prove that a classification scheme for the Higgs sector and for the YM part with a fixed spinorial background follows from a classification for the YM part (see Sections \ref{sec_higgs} and \ref{sec_further_examples}, respectively), so that the next step is to work on the classification of extensions of the spinorial sector for fixed YM background. 

In order to motivate our definition of extension, let us begin by noticing that in the current literature we find many ways to extend Yang-Mills theories, such as those in \cite{YM_deformation_1, YM_deformation_2, YM_deformation_CS, tensor_gauge_1, tensor_gauge_2, tensor_gauge_3, SW_map, forgacs1980, Harnad, PhysRevD.99.115026, YM_extension_book, YM_extension_families, YM_extension_stringy}. We note that most of them can be organized into three classes: \textit{deformations}, \textit{addition of a correction term} and \textit{extension of the gauge group}. In the first class there exists a theory $\hat{S}$ which is the limit $\lim_{i \rightarrow \infty} \hat{S}_i$ of a family of other theories $\hat{S}_i$ and such that the first term $\hat{S}_0$ on the sequence is precisely a YM theory $S^G$. Typically, the sequence is formed by partial sums, i.e,  $\hat{S}_i=\sum_{j\leq i}f_j(\lambda_j)S_j$ depending on certain \textit{fundamental} (or \textit{deformation}) parameters $\lambda_j$ and on functions of them such that $f_0(\lambda_0)=1$ and  $\lim_{\lambda_j\rightarrow 0}f_j=0$. In this case, $\hat{S}=\sum_{i\geq 0}f_i(\lambda_i)S_i$ and the YM theory $S^G$ is recovered in the limit $(\lambda_I)\rightarrow 0$, where $I=\{0,1,...$\}. In the second class of examples, we begin with a YM theory $S^G$ and we add some sort of correction term $C$, typically breaking some symmetry, so that the extended theory is the sum $\hat{S}=S^G+C$. In the third one there is a theory $\hat{S}^{\hat{G}}$ depending on a larger group $\hat{G}$, with $G\subseteq \hat{G}$, and such that there is some process of reduction from $\hat{G}$ to $G$ such that, when applied to $\hat{S}^{\hat{G}}$ we recover $S^G$. Typically $\hat{G}=G\times K$ for some other group $K$ and the reduction process is some dimensional reduction.

Now, let $\hat{S}$ be an extension by deformation (in the sense above) of a YM theory $S^G$ and note that $\hat{S}=S_0+\sum_{i\geq 1}f_i(\lambda_i)S_i $. Let $\delta$ be the map which assigns a configuration $\varphi(\lambda_I)$ in the domain of $\hat{S}$ to its limit $\lim_{\lambda_I \rightarrow0 }\varphi(\lambda_I)$. Thus, the composition $S^G\circ \delta$ is precisely $S_0$, so that one can write 
\begin{equation}\label{decompo_intro}
  \hat{S}=S^G\circ \delta + \sum_{i\geq 1}f_i(\lambda_i)S_i,  
\end{equation}
which basically says that, except for the presence of the $\delta$-map, an extension by deformation is an extension by adding the correction term $C=\sum_{i\geq 1}f_i(\lambda_i)S_i$. This suggests that there must be a wider notion of ``extension'' which unifies both classes, as we will propose in this paper.

Actually, we will work in a more general setting. Most of the discussion is about the geometry underlying the YM theories, so that it does not depends on the way used to attach this geometry in order to build the action functional. More precisely, our entire discussion applies for gauge theories whose Lagrangian is given by some quadratic form $q(F_D)$ of the curvature, not necessarily that given by $\operatorname{tr}(F_D\wedge \star F_D)$. This has the great advantage of avoiding the requirements on the existence of a semi-riemannian structure in the spacetime (used to build the Hodge star) and on a compactness and/or semi-simpleness hypothesis on the Lie algebra (needed to make the Killing form a nice pairing). Thus, throughout the paper we work with what we call \textit{Yang-Mills-type} (YMT) theories. As we show the additional degrees of freedom on the choice of the pairing have real meaning, in the sense that there are really many more YMT theories than YM theories.

Looking at the two classes of extensions discussed in the literature described above, one can see that deformations are about \textit{restriction of scales}, while addition of correction terms is about summing terms in the \textit{entire domain}. Thus, if we are searching for an unifying notion of extension this should be about \textit{adding correction terms in some scales allowing an enlarging of the gauge group}. This is the core of our definition. Indeed, given a YMT theory $S^G$ with gauge group $G$, an \textit{extension} for it is defined by the following:
\begin{enumerate}
    \item a possibly larger group $\hat{G}$ such that $G\subseteq \hat{G}$, representing the enlargement of the gauge group;
    \item a space $\widehat{\operatorname{Conn}}_{\hat{G}}$ containing all $\hat{G}$-connections and which is invariant by global gauge transformations of $\hat{G}$, called the \textit{extended domain} and describing the domain in which the extended theory is defined;
    \item a gauge invariant action functional $\hat{S}^{\hat{G}}:\widehat{\operatorname{Conn}}_{\hat{G}} \rightarrow \mathbb{R}$, called the \textit{extended functional} and defining the extended theory;
    \item a smaller space $C_{\hat{G}}\subset \widehat{\operatorname{Conn}}_{\hat{G}}$, called the \textit{correction domain} and playing the role of a special regime, or scale;
    \item another functional $C_{\hat{G}}\rightarrow  \mathbb{R}$, called the \textit{correction functional}, corresponding to the additional term;
    \item a map $\delta:C_{\hat{G}}\rightarrow \operatorname{Conn}_G$ connecting the configurations at a special scale with the configurations of the starting YMT theory,
\end{enumerate}
all of this subjected to the condition $\hat{S}^{\hat{G}} \vert_{C_{\hat{G}}}=S^G\circ \delta +C$, which is the analogue of (\ref{decompo_intro}). 

Once a class of objects is introduced, the main classification problem is about finding a bijective correspondence between this class and some other set which is known a priori. Thus, the primary classification problem for extensions of a YMT theory $S^G$ is about studying the space $\operatorname{Ext}(S^G;\hat{G})$ of its extensions relative to a fixed extended gauge group $\hat{G}$. We prove that if a group $\mathbb{G}$ acts in the set of real numbers $\mathbb{R}$ preserving the sum, then it induces a structure of a $R[\mathbb{G}]$-module bundle for any commutative unital ring R and we make the conjecture that for certain $R$ and $\mathbb{G}$ this bundle is a continuous subbundle of a specific trivial $R[\mathbb{G}]$-module bundle. Thus, in order to classify the extensions of $S^G$ one can analyze the algebraic-topological properties of the corresponding module bundles, which should be another natural next step in this program.

On the other hand, recall that there is another (more refined) way to work on the classification of a class of objects: by means of considering how each object interacts with each other. In other words, one can search for natural notions of ``morphisms'' between such objects in such a way that they constitute a category $\mathbb{C}$. In this case, one can try to classify the category itself or the starting class of objects, now up to isomorphisms. Following this philosophy, we prove that there really exists a notion of \textit{morphisms between extensions} so that $\operatorname{Ext}(S^G;\hat{G})$ is the collection of objects of a category $\mathbf{Ext}(S^G;\hat{G})$ of extensions of $S^G$. We embed this category in a conservative way into a slice category, which produces some constraints on the possible categorical constructions that can be done between extensions. We also conjecture that this category can be regarded, for certain additive $\mathbb{G}$-actions on $\mathbb{R}$ and certain rings $R$, a category internal to the category $\mathbf{Bun}_{R[\mathbb{G}]}$.

There is also a third approach, more constructive, to the classification of $\operatorname{Ext}(S^G;\hat{G})$. It is about first classifying more easy subclasses $\mathcal{E}(S^G;\hat{G})\subset \operatorname{Ext}(S^G;\hat{G})$. Notice that the categorical approach also applies to this case, since each subclass of the class of objects of a category $\mathbf{C}$ induces a full subcategory. In our context, notice that $\operatorname{Ext}(S^G;\hat{G})$ can be decomposed into two disjoint subclasses: those with vanishing correction term (i.e, such that $C\equiv 0$) and those with non-null correction term. The first one is the class of the so-called \textit{complete extensions}, while the second one are the \textit{incomplete extensions}. In \cite{complete_extension} we propose a list of four additional problems which should be studied for a given subclass $\mathcal{E}(S^G;\hat{G})\subset \operatorname{Ext}(S^G;\hat{G})$ in view to the classification problem: existence problem, universality problem, maximality problem and universality problem and we show that they have solutions internal to any ``coherent'' class of complete extensions. This proves, in particular, that well-behaved complete extensions exist with some generality. The study of incomplete extensions is made in \cite{gauge_breaking_YMT_extensions}, where it is shown that equivariant extensions are always complete, but there is a canonical class of incomplete equivariant extensions and even incomplete gauge-breaking extensions, suggesting that the class of incomplete extensions is more nasty.

Let us finish this introduction with a brief description of how the paper is organized. In Section \ref{sec_YMT} the notion of Yang-Mills-type theory is formally defined and we study the space of all of them. We prove that they define a nontrivial fibration over the set of all triples $(M,G,P)$, where $M$ is the manifold, $G$ a Lie group and $P$ a principal $G$-bundle. We show that if $\dim G\geq 2$ and if $G$ is not discrete, then the corresponding fibers are infinite-dimensional as real vector spaces, but finitely-generated as $C^{\infty}(M)$-modules and in this case we provide an upper bound for their rank. In Section \ref{sec_extensions} the concept of extension is introduced in more precise terms and we explicitly show that the most notions of extensions arising in the literature are particular examples of ours. We also show that the emergence phenomena (in the sense of \cite{emergence_yuri}) are a source of examples for nontrivial extensions.  In Section \ref{sec_space_extensions} the properties of the space of extensions $\operatorname{Ext}(S^G;\hat{G})$ described above are proved. Finally, in Section \ref{sec_category_extensions} the category of extensions $\mathbf{Ext}(S^G;\hat{G})$ is introduced and proved to be a conservative subcategory of a product of slice categories.

\section{Yang-Mills-Type Theories}\label{sec_YMT}
\label{sec:classicalcase}

\quad\;\,We begin by recalling that a \emph{Yang-Mills theory} (YM) is given by
\begin{enumerate}
    \item a $n$-dimensional compact, orientable semi-Riemannian smooth manifold $(M,g)$, regarded as the spacetime;
    \item a real or complex finite-dimensional Lie group $G$, regarded as the gauge group of internal symmetries;
    \item a principal $G$-bundle $P$ over $M$, called the \emph{instanton sector} of the theory.
\end{enumerate}
\quad\;\, The \textit{action functional} is the map $S:\operatorname{Conn}(P;\mathfrak{g})\rightarrow \mathbb{R}$, defined on the space of $G$-connections on $P$, given by
    \begin{equation}\label{YM_action}
        S[D]:= \int_{M} \langle F_D, F_D\rangle_{\mathfrak{g}} dvol_g.
    \end{equation}
Here, $dvol_g$ is the volume form induced by $g$, $F_D=dD+\frac{1}{2}D[\wedge]_{\mathfrak{g}}D$ is the field strength of $D$ and $[\wedge]_{\mathfrak{g}}$ is the wedge product induced by the Lie bracket of $\mathfrak{g}$ (we are using the notations and conventions of \cite{EHP_yuri_rodney,obstructions_yuri_rodney}). Furthermore, $\langle\cdot,\cdot\rangle_{\mathfrak{g}}: \Omega_{heq}^2(P;\mathfrak{g})\otimes \Omega_{heq}^2(P;\mathfrak{g})\rightarrow C^{\infty}(M)$ is the standard $C^{\infty}(M)$-linear pairing on the space of horizontal $G$-equivariant $\mathfrak{g}$-valued 2-forms on $P$, whose definition we recall briefly (see \cite{advanced_field_theory,donaldson_livro} for further details). The semi-riemannian metric $g$ induces a pairing
$\langle\cdot,\cdot\rangle_{g}:\Omega^{2}(M)\otimes\Omega^{2}(M)\rightarrow C^{\infty}(M)$,
given by $\langle\alpha,\beta\rangle_{g}=\alpha\wedge\star_{g}\beta$.
On the other hand, the Killing form of $B_{\mathfrak{g}}:\mathfrak{g}\otimes\mathfrak{g}\rightarrow\mathbb{R}$
of $\mathfrak{g}$ extends to a pairing $\overline{B}_{\mathfrak{g}}:\Gamma(E_{\mathfrak{g}})\otimes\Gamma(E_{\mathfrak{g}})\rightarrow C^{\infty}(M)$
on the global sections of the adjoint bundle $E_{\mathfrak{g}}=P\times_{M}\mathfrak{g}$.
Since $\Omega^{2}(M;E_{\mathfrak{g}})\simeq\Omega^{2}(M)\otimes\Gamma(E_{\mathfrak{g}})$,
together the pairings $\langle\cdot,\cdot\rangle_{g}$ and $\overline{B}_{\mathfrak{g}}$
induce a pairing $\langle\cdot,\cdot\rangle_{g}\otimes\overline{B}_{\mathfrak{g}}\equiv\langle\cdot,\cdot\rangle_{g,\mathfrak{g}}$
on $\Omega^{2}(M;E_{\mathfrak{g}})$. But we have a canonical isomorphism
$\tau:\Omega_{heq}^{2}(P;\mathfrak{g})\simeq\Omega^{2}(M;E_{\mathfrak{g}})$,
leading us to take the pullback of $\langle\cdot,\cdot\rangle_{g,\mathfrak{g}}$
and define the desired pairing $\langle\cdot,\cdot\rangle_{\mathfrak{g}}$
by $\langle\omega,\omega'\rangle_{\mathfrak{g}}=\langle\tau(\omega),\tau(\omega')\rangle_{g,\mathfrak{g}}$.
In the literature this is typically written as $\langle\omega,\omega'\rangle_{\mathfrak{g}}=\operatorname{tr}(\omega[\wedge]_\mathfrak{g}\star_g\omega')$.  We also recall that the functional (\ref{YM_action}) is invariant by the group $\operatorname{Gau}_G(P)$ of global gauge transformations of $P$, i.e, the group of its $G$-principal $M$-bundle automorphisms (we are using notations and conventions of \cite{michor_global_analysis}).

Throughout this paper we will work with \emph{Yang-Mills-type} (YMT) \emph{theories} which differ from the YM theories by replacing the standard pairing $\langle\cdot,\cdot\rangle_{\mathfrak{g}}$ with an arbitrary (possibly degenerate and non-symmetric) $\mathbb{R}$-linear pairing $\langle\cdot,\cdot\rangle: \Omega_{heq}^2(P,\mathfrak{g})\otimes \Omega_{heq}^2(P,\mathfrak{g})\rightarrow C^{\infty}(M)$. The \emph{action functional} of a YMT theory is defined analogously to (\ref{YM_action}): it is given the map $S:\operatorname{Conn}(P;\mathfrak{g})\rightarrow \mathbb{R}$ such that
    \begin{equation}\label{YMT_action}
        S[D]:= \int_{M}\langle F_D,F_D\rangle dvol_g.
    \end{equation} 
\begin{remark}
We emphasize that in a YMT theory we assume that the pairing is only $\mathbb{R}$-linear. If the pairing is actually $C^{\infty}(M)$-linear we will say that the corresponding YMT is \textit{tensorial} or \textit{linear}. This additional assumption is important when one does local computations, as in \cite{gauge_breaking_YMT_extensions}. Since all results of this paper do not involve such local computations, we will work in the general $\mathbb{R}$-linear setting unless explicit stated otherwise.
\end{remark}
\begin{remark}
Since the Killing form is an invariant polynomial, it follows that the classical YM theories are invariant by the action of the group $\operatorname{Gau}_G(P)$ of global gauge transformations.We also emphasize that we will \textit{avoid} the assumption of invariance under global gauge transformations on the pairings $\langle \cdot , \cdot \rangle$, i.e, we will \textit{not} require that $\langle f^{*}\left(F_D\right),f^{*}\left(F_D\right)\rangle = \langle F_D,F_D\rangle$ for every $f\in \operatorname{Gau}_{G}(P)$. This means that for us, \textit{a YMT theory can be gauge-breaking}. This generality is important when studying gauge-breaking extensions of YMT theories \cite{gauge_breaking_YMT_extensions}. In this work, when a YMT theory is such that his pairing is gauge invariant we will refer to it explicitly as a \textit{gauge invariant YMT theory}.  
\end{remark}
\begin{remark}
After standard modifications (replacing arbitrary smooth functions by densities), all definitions above make sense for noncompact and non-orientable smooth manifolds. However, for simplicity we will work in the compact and oriented case.
\end{remark}
\subsection{Yang-Mills \textit{vs} Yang-Mills-Type}\label{sec_YM_vs_YMT}

\quad\;\,In this subsection we will compare the concepts of YM and YMT theories. Let $\mathcal{X}$ be the
space of all triples $(M,G,P)$, where $M$ is a (compact and oriented)
semi-riemannian manifold, $G$ is a finite-dimensional Lie group and
$P$ is an isomorphism class of principal $G$-bundle on $M$. Furthermore,
let $\mathcal{YM}$, $\mathcal{YMT}$ and $\mathcal{YMT}_0$  be the spaces of all Yang-Mills,
Yang-Mills-type and linear Yang-Mills theories respectively. We have obvious projections
$\pi_{YM}:\mathcal{YM}\rightarrow\mathcal{X}$, $\pi_{YMT}:\mathcal{YMT}\rightarrow\mathcal{X}$ and $\pi_{YMT,0}:\mathcal{YMT}_0\rightarrow\mathcal{X}$.
Fixed $(M,G,P)\in\mathcal{X}$, let $\operatorname{YM}_{G}(P)$, 
$\operatorname{YMT}_{G}(P)$ and $\operatorname{YMT}_{G,0}(P)$ denote the corresponding fibers by $\pi_{YM}$,
$\pi_{YMT}$ and $\pi_{YMT,0}$, respectively. They are in bijection with the possible
pairings on each context. Thus, $\operatorname{YM}_{G}(P)$ has a
single object for every $(M,G,P)$, characterized by the canonical
pairing $\langle\omega,\omega'\rangle_{\mathfrak{g}}=\operatorname{tr}\omega\wedge\star_g\omega'$.
In particular, $\pi_{YM}:\mathcal{YM}\simeq\mathcal{X}$ is a bijection.
On the other hand, as consequence of the next lemma we will see that
$\operatorname{YMT}_{G}(P)$ and $\operatorname{YMT}_{G,0}(P)$ typically have a lot of elements, so that
$\pi_{YMT}:\mathcal{YMT}\rightarrow\mathcal{X}$ and $\pi_{YMT,0}:\mathcal{YMT}_0\rightarrow\mathcal{X}$ are nontrivial fibrations.

Let $R$ be a commutative ring, $\mathbb{K}\subset R$ be a subring which is also a field, $V$ and $Z$ two $R$-modules and $T:Z_{\mathbb{K}}\rightarrow\mathbb{K}$ a $\mathbb{K}$-linear map, where $Z_{\mathbb{K}}$ denotes the restriction of scalars.
A \emph{$(T,\mathbb{K})$-pairing} (resp. \emph{$(T,R)$-pairing})  on $V$  is a $\mathbb{K}$-linear map $B:V_{\mathbb{K}}\otimes_{\mathbb{K}}V_{\mathbb{K}}\rightarrow\mathbb{K}$
which factors through $T$, that is, such that $B=T\circ \overline{B}$ (resp. $B=T\circ \Tilde{B}_{\mathbb{K}}$) for some 
$\mathbb{K}$-linear map (resp. $R$-linear map) $ \overline{B}:V_{\mathbb{K}}\otimes_{\mathbb{K}} V_{\mathbb{K}} \rightarrow Z_{\mathbb{K}}$ (resp. $(\Tilde{B}:V\otimes_R V\rightarrow Z)$), where $\Tilde{B}_{\mathbb{K}}$ is the scalar restriction of $\Tilde{B}$.
Let $\operatorname{Pair}_{T,\mathbb{K}}(V)$ (resp. $\operatorname{Pair}_{T,R}(V)$) 
denote the space of all of them. Furthermore, let $\operatorname{Pair}_R(V)$ be the $R$-module of all $R$-linear maps $T:V\otimes_R V\rightarrow R$.

\begin{lemma}
In the same notations above, let $V$, $W$ and $Z$ be $R$-modules and let $T:Z_\mathbb{K} \rightarrow\mathbb{K}$ be a $\mathbb{K}$-linear map. If $T$ is injective, then we have isomorphisms of $\mathbb{K}$-vector spaces:
\begin{align}
 \operatorname{Pair}_{(T,\mathbb{K})}(V\otimes_R W)&\simeq\operatorname{Hom}_{\mathbb{K}}(V_{\mathbb{K}}^{\otimes_{\mathbb{K}}^2}\otimes_{\mathbb{K}} W_{\mathbb{K}}^{\otimes_{\mathbb{K}}^2};Z_{\mathbb{K}}) \\ \label{lemma_pairs_2}
\operatorname{Pair}_{(T,R)}(V\otimes_R W) &\simeq [\operatorname{Hom}_R(V^{\otimes^2_R}\otimes_R W^{\otimes^2_R};Z)]_{\mathbb{K}}
\end{align}
where the right-hand side is the space of linear maps and $X^{\otimes^n}=X\otimes ...\otimes X$, $n$-times. If in addition $V$, $W$ and $R$ are projectives as $R$-modules, then 
\begin{equation}\label{pairs_vs_tensor}
  \operatorname{Pair}_{(T,R)}(V\otimes_R W)\simeq [\operatorname{Pair}_R(V)\otimes_R\operatorname{Pair}_R(W)]_{\mathbb{K}}.  
\end{equation}
\end{lemma}
\begin{proof}
For the first case, notice that $\operatorname{Pair}_{T,\mathbb{K}}(V\otimes_{R}W)$
is precisely the image of the canonical map
\[
T_{*}:\operatorname{Hom}_{\mathbb{K}}((V_{\mathbb{K}}\otimes_{\mathbb{K}}W)^{\otimes ^2_{\mathbb{K}}};Z_{\mathbb{K}})\rightarrow \operatorname{Hom}_{\mathbb{K}}((V_{\mathbb{K}}\otimes_{\mathbb{K}}W)^{\otimes ^2_{\mathbb{K}}};\mathbb{K}),
\]
given by $T_*(\overline{B})=T\circ\overline{B}$. Since covariant hom-functors preserve monomorphisms, it follows that $T_{*}$ is injective. The result then follows from the isomorphism
theorem and from the commutativity up to isomorphisms of the tensor
product. For the second case, notice that for every $V,W$ and $Z$ we have a canonical $\mathbb{K}$-linear injective map 
$$
\alpha:\operatorname{Hom}_{\mathbb{K}}([(V\otimes_R W)^{\otimes^2_R}]_\mathbb{K};Z_{\mathbb{K}})\rightarrow \operatorname{Hom}_{\mathbb{K}}((V_{\mathbb{K}}\otimes_{\mathbb{K}}W)^{\otimes ^2_{\mathbb{K}}};Z_{\mathbb{K}}),
$$
obtained as follows. Recall that restriction of scalars is right-adjoint to extension of scalars. Since tensor product is commutative up to isomorphisms, the extension of scalars functor is a strong monoidal functor relative to the monoidal structure given by tensor products (Chapter II of \cite{bourbaki}), so that its left adjoint is lax monoidal (see Chapter 5 of \cite{monoidal_adjunction}). Therefore, for every $R$-modules $X,Y$ there is a $\mathbb{K}$-linear map $\mu: X_{\mathbb{K}}\otimes_{\mathbb{K}}Y_{\mathbb{K}}\rightarrow (X\otimes _R Y)_{\mathbb{K}}$ which in our case is a projection and then surjective. Thus, for $X=V\otimes_R W=Y$, a diagram chasing give us a surjective map 
$$
\mu:(V_{\mathbb{K}}\otimes_{\mathbb{K}}W)^{\otimes ^2_{\mathbb{K}}} \rightarrow [(V\otimes_R W)^{\otimes^2_R}]_\mathbb{K},
$$
so that $\alpha = \mu^*$, i.e, $\alpha (\varphi)=\varphi \circ \mu$. Finally, observe that $\operatorname{Pair}_{T,R}(V\otimes_{R}W)$ is the image of the following composition with $X=V\otimes_R W=Y$:
$$
\xymatrix{\ar@{-->}[d] \operatorname{Hom}_R (X\otimes_R Y;Z) \ar[r]^-{(-)_{\mathbb{K}}} & \operatorname{Hom}_{\mathbb{K}} ([X\otimes_R Y]_{\mathbb{K}};Z_{\mathbb{K}}) \ar[r]^-{\alpha} & \operatorname{Hom}_{\mathbb{K}}(X_{\mathbb{K}}\otimes _{\mathbb{K}} Y_{\mathbb{K}};Z_{\mathbb{K}}) \ar@/^/[lld]^{T_{*}} \\
\operatorname{Hom}_{\mathbb{K}}(X_{\mathbb{K}}\otimes _{\mathbb{K}} Y_{\mathbb{K}};\mathbb{K})}
$$
Since restriction of scalars is a faithful functor, the first of these maps is injective \cite{riehl_category}. The second one is injective from the above disussion, while the third one is injective by the first case of this lemma, so that again by the isomorphism theorem and the commutativity up to isomorphisms of the tensor product we get the desired result. Now, for (\ref{pairs_vs_tensor}) recall that we always have a canonical morphism $$\operatorname{Hom}_R(V;R)\otimes_R\operatorname{Hom}_R(W;R)\rightarrow \operatorname{Hom}_R(V\otimes_RW;R)$$ which is an isomorphism if the $R$-modules in question are projective \cite{bourbaki}. Composing this isomorphism with (\ref{lemma_pairs_2}) we get (\ref{pairs_vs_tensor}). 
\end{proof}

\begin{theorem}
If $G$ is not discrete and $\dim M\geq2$, then for every bundle
$P$ the fibers $\operatorname{YMT}_{G}(P)$ and $\operatorname{YMT}_{G,0}(P)$ of $\pi_{YMT}:\mathcal{YMT}\rightarrow\mathcal{X}$ and $\pi_{YMT,0}:\mathcal{YMT}_0\rightarrow\mathcal{X}$, respectively, are infinite-dimensional real vector spaces. Furthermore, as a $C^{\infty}(M)$-module   $\operatorname{YMT}_{G,0}(P)$ is projective of finite rank. Otherwise, i.e, if $G$ is discrete and/or $\dim M=0,1$, then the fibers are zero-dimensional as real vector spaces.
\end{theorem}
\begin{proof}
First of all, notice that 
\begin{align}
    \operatorname{YMT}_{G}(P)&\simeq \operatorname{Pair}_{(\int_M, \mathbb{R})}(\Omega^2(M)\otimes_{C^{\infty}} \Gamma(E_\mathfrak{g})) \\
     \operatorname{YMT}_{G,0}(P)&\simeq \operatorname{Pair}_{(\int_M, C^{\infty})}(\Omega^2(M)\otimes_{C^{\infty}} \Gamma(E_\mathfrak{g})), 
\end{align}
where $\int_M:C^\infty(M)_{\mathbb{R}}\rightarrow \mathbb{R}$ is the integral and we wrote $C^{\infty}$ instead of $C^{\infty}(M)$ in order to simplify the notation. Since the integral is injective up to sets with zero measure, a small change in the previous lemma allows us to conclude that 
\begin{align}
\operatorname{YMT}_{G}(P)&\simeq\operatorname{Hom}_{\mathbb{R}}((\Omega^{2}(M)_\mathbb{R}^{\otimes_\mathbb{R}^2} \otimes_{\mathbb{R}} \Gamma(E_{\mathfrak{g}})_\mathbb{R}^{\otimes_\mathbb{R}^2}  ;C^{\infty}(M)_{\mathbb{R}})\label{dimension_YMT}\\
\operatorname{YMT}_{G,0}(P)&\simeq\operatorname{Hom}_{C^{\infty}}((\Omega^{2}(M)^{\otimes_{C^{\infty}}^2} \otimes_{C^{\infty}} \Gamma(E_{\mathfrak{g}})^{\otimes_{C^{\infty}}^2}  ;C^{\infty}(M))_{\mathbb{R}}\label{dimension_YMT_0}
\end{align}
Thus, if $G$ is discrete, then $\dim G=0$, so that $\dim\mathfrak{g}=0$
and $E_{\mathfrak{g}}\simeq M\times0$, implying $\Gamma(E_{\mathfrak{g}})_{\mathbb{R}}\simeq0$,
which means that the whole (\ref{dimension_YMT}) is zero-dimensional. Since the rank of a $C^\infty(M)$-module is bounded from above by its rank as a $\mathbb{R}$-module, it follows that under the previous assumptions (\ref{dimension_YMT_0}) is zero-dimensional too. 
Similarly, if $\dim M<2$, then $\Omega^{2}(M)_\mathbb{R}\simeq0$, showing again
that (\ref{dimension_YMT}) and (\ref{dimension_YMT_0}) are zero-dimensional. Thus, suppose that
$\dim M\geq2$ and that $G$ is not discrete. In this case, the bundles $E_{\mathfrak{g}}$, $\Lambda ^2 T^*M$ have positive rank. But if $E\rightarrow M$ is any vector bundle with positive rank, then $\dim_{\mathbb{R}}(\Gamma (E)_\mathbb{R})=\infty$, so
that $\operatorname{YMT}_{G}(P)$ is infinite-dimensional. Notice that
\begin{equation}\label{YMT_0_as_module}
\operatorname{Hom}_{C^{\infty}}((\Omega^{2}(M)^{\otimes_{C^{\infty}}^2} \otimes_{C^{\infty}} \Gamma(E_{\mathfrak{g}})^{\otimes_{C^{\infty}}^2}  ;C^{\infty}(M))\simeq \Gamma (\operatorname{Hom}(\Lambda ^2 T^{*}M^{\otimes^2}\otimes E_{\mathfrak{g}}^{\otimes^2}; M\times \mathbb{R})),    
\end{equation}
so that $\operatorname{YMT}_{G,0}(P)$ is also infinite-dimensional as a real vector space.  Finally, recall that, regarded as a $C^{\infty}(M)$-module, the space of global sections $\Gamma(E)$ of any vector bundle is projective and of finite rank. Thus, from the isomorphism above, we see that under the hypotheses $\operatorname{YMT}_{G,0}(P)$ is a projective $C^{\infty}(M)$-module of finite rank.
\end{proof}
\begin{corollary}
For discrete gauge groups (and therefore for instanton sectors given
by covering spaces) Yang-Mills-type theories, linear Yang-Mills-type theories and Yang-Mill theories are all the same thing.
\end{corollary}
\begin{proof}
By the last theorem, if $G$ is discrete, then $\operatorname{YMT}_G(P)\simeq 0 \simeq \operatorname{YMT}_{G,0}(P)$. But from the discussion at the beginning of Subsection \ref{sec_YM_vs_YMT} we know that  $\operatorname{YM}_G(P)$ has always a single element. Thus, under the hypothesis all three sets are in bijection.
\end{proof}

\subsection{Upper Bound}

\quad \;\,If $\dim M\geq 2$ and $G$ is not discrete, we can actually give an upper bound to rank of  $\operatorname{YMT}_{G,0}(P)$ when regarded as a $C^{\infty}(M)$-module. This goes as follows.
By the Serre-Swan theorem (and its extensions to the non-compact case), the category of projective finitely generated $C^{\infty}(M)$-modules is equivalent to the category of vector bundles with finite rank over $M$, and the equivalence is given precisely by the functor of global sections \cite{serre_swan}. Thus, if $E$ is a bundle such that $\Gamma (E)$ is generated by $N$ elements, then the module $\Gamma (E)$ has rank bounded from above by $N$, i.e, $\operatorname{rnk}\Gamma (E)\leq N$. 

But in the construction of this generating set, we see that $N=m+k$, where $k$ is the number of elements of a finite trivializing open covering for $E$. Let $\min_k(E)$ be the minimum of such $k$, i.e, the minimum number of elements in a trivilizing open covering of $E$. Thus, $1\leq \min_k(E)$ and $\min_k(E)=1$ if $E$ is trivial. It can be proved that for $\min_k(E)\leq \operatorname{cat}(M)$ for every $E$, where $\operatorname{cat}(M)$ is the Lusternik-Schnirelmann category of $M$, which satisfies $\operatorname{cat}(M)\leq n+1$, where $n$ is the dimension of $M$ \cite{LS_category}. Thus, for every vector bundle $E$ we have the upper bound $\operatorname{rnk}\Gamma (E)\leq m+n+1$. One can also prove that if $M$ is $q$-connected with i.e, $\pi_i(M)=0$ for $1\leq i\leq q$, then $\operatorname{cat}(M)< (n+1/q+1)+1$, so that $\operatorname{rnk}\Gamma (E)<m+(n+1/q+1)+1$.

\begin{remark}\label{remark_rank}
We notice that if $M$ is not contractible, then $q < n$. Indeed, every topological $n$-manifold is homotopic to a $n$-dimensional CW-complex \cite{milnor_CW,wall1_CW}. But if a $n$-dimensional CW-complex $X$ is $q$-connected, with $q\geq n$, then $X$ is contractible.  On the other hand, if $M$ is contractible, then $M$ is not compact\footnote{Recall that a compact manifold of positive dimension is never contractible.} and every bundle $E\rightarrow M$ is trivial. Thus, $
\operatorname{rank}(E)=m+1$.
\end{remark}

\begin{corollary}
For every $l$-dimensional Lie group $G$, with $l>0$, every smooth $n$-dimensional manifold $M$, with $n\geq 2$ and every principal $G$-bundle $P$ we have $$
\operatorname{rnk}\operatorname{YMT}_{G,0}(P)\leq \frac{(n ^2 - n)^2l^2}{4}+n +1.\quad
$$
Furthermore,
\begin{enumerate}
    \item if the smooth manifold $M$ is $q$-connected, we also have 
$$
\operatorname{rnk}\operatorname{YMT}_{G,0}(P)< \frac{(n ^2 - n)^2l^2}{4}+\frac{n+1}{q+1} +1;
$$
\item if $M$ is paralellizable and $G$ is abelian, or if $M$ is contractible, then 
\begin{equation}\label{rank_contractible}
  \operatorname{rnk}\operatorname{YMT}_{G,0}(P)= \frac{(n ^2 - n)^2 l^2}{4} +1.  
\end{equation}
\end{enumerate}

\end{corollary}
\begin{proof}
For the main assertion and for (1), just apply the previous bounds to (\ref{YMT_0_as_module}) noticing that the rank of the bundle 
\begin{equation} \label{bundle}
\operatorname{Hom}(\Lambda ^2 T^{*}M^{\otimes^2}\otimes E_{\mathfrak{g}}^{\otimes^2}; M\times \mathbb{R})
\end{equation}
is given by $(n^2-n)^2l^2/4$. For the case (2), notice that if $M$ is paralellizable, then $\Lambda^2 T^*M$ is trivial. Furthermore, the adjoint bundle of a principal $G$-bundle with abelian $G$ is also trivial. Thus, (\ref{bundle}) is trivial and the result follows from the previous discussion. If, instead, $M$ is contractible, then the bundle (\ref{bundle}) is automatically trivial and the proof is done.   
\end{proof}

Some interesting situations to keep in mind:

\begin{example}[lower dimension examples] \emph{As we saw above, the lowest rank of $\operatorname{YMT}_{G,0}(M)$ is realized when $M$ is contractible or when $M$ is parallelizable and $G$ is abelian. Assuming one of these conditions and fixing a upper bound $z$ for (\ref{rank_contractible}) and find the integer solutions for the inequality 
$$
\frac{(n^2-n)^2l^2}{4} +1 \leq z,
$$
where $n\geq 2$ and $l>0$. For instance, if we fixes $z=7$, we get the graphs below. In particular, we see that \textit{there is no triple $(M,G,P)\in \mathcal{X}$ such that $\operatorname{rnk}\operatorname{YMT}_{G,0}(M)=1$}. } 
\begin{figure}[H]
\begin{centering}
\includegraphics[scale=0.35]{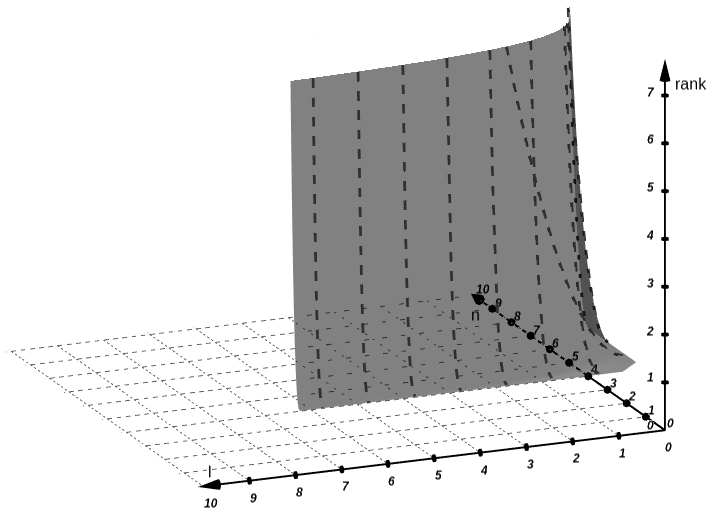}\includegraphics[scale=0.35]{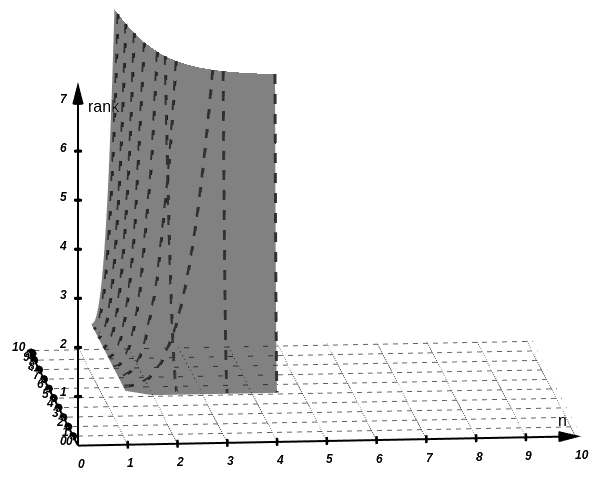}
\par\end{centering}
\caption{Two different views of the graphic for the values of $n$ and $l$ which produces a space of linear YMT theories
with rank $z\leq7$. }
\end{figure}
\end{example}

\subsection{Gauge Invariant Case}
\quad\;\,This section is a remark on the classification of the subspace $\operatorname{YMT}_{G,0}(P)_{\operatorname{gau}}\subset \operatorname{YMT}_{G,0}(P)$ of gauge invariant linear YMT theories with prescribed Lie group $G$ and instanton sector $P$. Since $\operatorname{YMT}_{G,0}(P)_{\operatorname{gau}}\subset \operatorname{YMT}_{G}(P)_{\operatorname{gau}}\subset \operatorname{YMT}_{G,0}(P)$, this can also be understood as lower bounds for the space $\operatorname{YMT}_{G}(P)_{\operatorname{gau}}$ of not necessaily linear gauge invariant YMT theories.
We begin by noting that from (\ref{dimension_YMT_0}) and (\ref{pairs_vs_tensor}) we get 
\begin{align}
 \operatorname{YMT}_{G,0}(P)&= \operatorname{Pair}_{C^{\infty}}(\Omega^2_{heq}{(P;\mathfrak{g}})) \\
 &\simeq \operatorname{Pair}_{C^{\infty}}((\Omega^{2}(M))\otimes_{C^{\infty}} \operatorname{Pair}_{C^{\infty}}( \Gamma(E_{\mathfrak{g}})) \\
 &\simeq \operatorname{Pair}_{C^{\infty}}((\Omega^{2}(M))\otimes_{C^{\infty}} \Gamma(\operatorname{Pair}(E_{\mathfrak{g}})),
\end{align}
where $\operatorname{Pair}(E)$ is the \textit{bundle of pairs}, i.e, the bundle $\operatorname{Hom}(E\otimes E;M\times \mathbb{R})$. Since $\Omega^2(M)$ is independent of $G$, the action by pullbacks of $\operatorname{Bun}_G(P)$ on $\operatorname{Pair}_{C^{\infty}}(\Omega^2_{heq}{(P;\mathfrak{g}}))$ corresponds, via the decomposition above, to an adjoint  action on $\Gamma(\operatorname{Pair}(E_{\mathfrak{g}}))$, i.e, $[f^*\langle\cdot, \cdot\rangle(s,s')]_p=\langle Ad_{f(p)}s(p), Ad_{f(p)}s'(p) \rangle $, where we are using the identification $\operatorname{Gau}_G(P)\simeq C^{\infty}_{eq}(M;G)$ of global gauge transformations with equivariant smooth maps. Let us define an \textit{adjoint structure} on $P$ as a $C^{\infty}$-linear pairing $\langle\cdot,\cdot \rangle$ on $E_{\mathfrak{g}}$ which is $\operatorname{Gau}_G(P)$-invariant, i.e, such that  $ \langle Ad_{f(p)}s(p), Ad_{f(p)}s'(p) \rangle =\langle s(p);s'(p)\rangle $ for every $p\in M$ and every sections $s,s'\in \Gamma (E_{\mathfrak{g}})$. Let $\operatorname{AdStr}(P;\mathfrak{g})$ be the set of all of them. It is actually a $C^{\infty}(M)$-submodule of $\Gamma(\operatorname{Pair}(E_{\mathfrak{g}}))$. Thus:

\begin{lemma}\label{lemma_YMT_gauge}
For every Lie group $G$ and every principal $G$-bundle $P$, we have $$\operatorname{YMT}_{G,0}(P)_{\operatorname{gau}}\simeq \operatorname{Pair}_{C^{\infty}}((\Omega^{2}(M))\otimes_{C^\infty}\operatorname{AdStr}(P;\mathfrak{g}).$$
\end{lemma}

Let $\operatorname{inv}(\mathfrak{g})$ be the semifree differential graded algebra of invariant polynomials of $\mathfrak{g}$. Every element of $\omega \in \operatorname{inv}^2(\mathfrak{g})$ degree 2 in $\operatorname{inv}(\mathfrak{g})$ induces an Ad-invariant pairing in $\mathfrak{g}$, which in turn induces an adjoint structure on $P$. Note that $\operatorname{inv}^2(\mathfrak{g})$ always contains the Killing form of $\mathfrak{g}$, but it is typically higher dimensional. For instance, even if we are looking for nondegenerate Ad-invariant pairings, one can find them not only on semi-simple Lie algebras (given in this case by the Killing form), but also on a large class of solvable Lie algebras \cite{ad_metrics}.

\section{Extensions}\label{sec_extensions}

\quad\;\, Let $G$ be a Lie group. A \emph{basic extension} of $G$ is another Lie group $\hat{G}$ and an injective Lie group homomorphism $\imath:G\rightarrow \hat{G}$. A basic extension is called a \emph{split} if there exists a Lie group isomorphism $\hat{G}\simeq \operatorname{Ker}(\imath)\ltimes G$. The relation between basic extensions and the classical notion of group extensions is as follows. Recall that a group $\hat{G}$ is called an \emph{extension} of another group $G$ if they belong to an short exact sequence of groups of the form
\begin{equation}
\label{diagram2}
    \xymatrix{1\ar[r]& N\ar[r]^{\imath}& \hat{G}\ar[r]^{\pi}&G\ar[r]& 1.}
\end{equation}

This implies that $N$ is isomorphic to the normal subgroup $i(N)\subset G$ and that $G$ is isomorphic to the group $\hat{G}\slash \imath(N)$. Hence $G$ is not necessarily a subgroup of $\hat{G}$ \textit{a priori}. It is precisely if the map $\pi:\hat{G}\rightarrow G$ has a global section, i.e. iff the extension is split. On the other hand an inclusion $i:G\rightarrow \hat{G}$ does not necessarily fulfill $\hat{G}$ as a classical extension of $G$, but it does iff $G$ admits a semidirect complement in $\hat{G}$ and in this case the induced extension is necessarily split. More precisely, we have the dotted horizontal arrow below and its image is precisely the image of the vertical arrow. Thus, it factors producing the diagonal one, which is a bijection.
\begin{equation}
\label{diagram3}
    \xymatrix{\{\text{split extensions\}}\ar@{^(-->}[r]\ar@{..>}[rd]_{\simeq}& \text{\{basic extensions\}}\\ &
    \text{\{split basic extension\}}\ar@{^{(}->}[u]}
\end{equation}
\quad\;\,In sum, we have the following conlusion:
\begin{conclusion}
The typical examples of basic extensions of a Lie group are their split extensions. Actually, split extensions are in bijection with split basic extensions.
\end{conclusion}
Let $P\rightarrow M$ be a principal $G$-bundle over a smooth manifold
$M$ and let $\imath : G\hookrightarrow\hat{G}$ be a basic extension
of $G$ (for example, a split extension by the above). From the classification theorem of principal bundles, $P$
is classified by a homotopy class of continuous maps $f:M\rightarrow BG$,
where $BG$ is the classifying space of $G$ \cite{michor_global_analysis,michor1988gauge}. Since $B$ is functorial,
we have a map $Bf:BG\rightarrow B\hat{G}$, which classifies a principal
$\hat{G}$-bundle $\hat{P}$. In terms of cocycles this means that
for every trivializing open covering $U_{i}$ of $M$, the cocycles
$\hat{g}_{ij}:U_{ij}\rightarrow\hat{G}$ are just the compositions
$\imath\circ g_{ij}$. Since the classification of bundles is realized
by pullbacks, from universality there exists a bundle map $\imath:P\hookrightarrow\hat{P}$, denoted by the same notation of the basic extension map, which is a monomorphism due to the stability of monomorphisms under
pullbacks.

Furthermore, each basic extension also induces a group homomorphism $\xi:\operatorname{Gau}_{G}(P)\rightarrow \operatorname{Gau}_{\hat{G}}(\hat{P})$. Indeed, recall that $\operatorname{Gau}_{G}(P)$ can be regarded as the group of global sections of the group bundle $P_G=(P\times G)/G\rightarrow M$ \cite{advanced_field_theory}. Since $\hat{P}$ is the bundle induced by the basic extension, the actions of $G$ on $P$ and of $\hat{G}$ on $\hat{P}$ are compatible, so that there exists the dotted arrow below, which is a morphism of group bundles. Applying the functor of global sections we get $\xi$. A similar homotopical conclusion is the following. In \cite{xi_homotopy} it is shown that we have a homotopy equivalence $\operatorname{Gau}_G(P)\simeq \Omega\operatorname{Map}(M;BG)_f$, where $f:M\rightarrow BR$ is the map whose homotopy class classifies $P$. Furthermore, $\Omega$ denotes the loop space, $\operatorname{Map}(X;Y)$ is the space of continuous maps in the compact-open topology and $\operatorname{Map}(X;Y)_f$ is path-component of $f$.    
$$
\xymatrix{\ar[d]_{\pi} P\times G \ar[rr]^{\imath \times \imath} && \hat{P} \times \hat{G} \ar[d]^{\hat{\pi}} \\
(P\times G)/G \ar@{-->}[rr] && (\hat{P} \times \hat{G})/\hat{G}}
$$

We can now define the main objects of this paper: \textit{extensions of YMT theories}. Let $G$ be a Lie group and let $S^G:\operatorname{Conn}(P;\mathfrak{g})\rightarrow \mathbb{R}$ be the action functional of a YMT theory with gauge group $G$ and instanton sector $P$. An \emph{extension} of $S^G$ consists of the following data:

\begin{enumerate}
    \item a basic extension $\imath: G \hookrightarrow \hat{G}$ of $G$;
    \item an equivariant extension of the space of connection 1-forms on $\hat{P}$. More precisely, a subset \begin{equation}{\label{domain_extension}
    \operatorname{Conn}(\hat{P};\hat{\mathfrak{g}}) \subseteq \widehat{\operatorname{Conn}}(\hat{P};\hat{\mathfrak{g}})\subseteq \Omega^{1}_{eq}(\hat{P};\hat{\mathfrak{g}}), }
    \end{equation}
    called the \textit{extended domain}, where $\hat{P}$ is the $\hat{G}$-bundle induced from $P$, which is invariant by the canonical action of $\operatorname{Gau}_{\hat{G}}(\hat{P})$ by pullbacks\footnote{Notice that although the YMT functionals are \textit{not} necessarily gauge invariant, their extended functional $\hat{S}^{\hat{G}}$ \textit{must be} gauge invariant};
    \item a $\operatorname{Gau}_{\hat{G}}(\hat{P})$-invariant functional $\hat{S}^{\hat{G}}:\widehat{\operatorname{Conn}}(\hat{P};\hat{\mathfrak{g}})\rightarrow \mathbb{R}$, called the \emph{extended action functional};
    \item a nonempty subset  $0\subseteq C^1(\hat{P};\hat{\mathfrak{g}})\subseteq \widehat{\operatorname{Conn}}(\hat{P};\hat{\mathfrak{g}})$, called the  \emph{correction subspace} and whose inclusion map in $\widehat{\operatorname{Conn}}(\hat{P};\hat{\mathfrak{g}})$ we will denote by $\jmath$;
    \item a \emph{correction functional} $C:C^1(\hat{P};\hat{\mathfrak{g}})\rightarrow \mathbb{R}$ and map $\delta: C^1(\hat{P};\hat{\mathfrak{g}})\rightarrow \operatorname{Conn}(P;\mathfrak{g})$ such that
    \begin{equation}\label{decomposition_extension}
      \hat{S}^{\hat{G}}\circ \jmath = S^{G}\circ \delta + C.
    \end{equation}
\end{enumerate}

Thus, if  $\hat{S}^{\hat{G}}$ is an extension of a YMT theory $S^{G}$ we have the diagram below, representing the sequence of subspaces and maps defining it. If $\widehat{\operatorname{Conn}}(\hat{P};\hat{\mathfrak{g}})\simeq \Omega^1_{eq}(\hat{P};\hat{\mathfrak{g}})$ we say that the extension is \emph{full}. If the correction functional $C$ is null, we say that the extension is \emph{complete on the correction subspace $C^1(\hat{P};\hat{\mathfrak{g}})$} or simply that it is \emph{complete}. If such subspace is such that $C^1(\hat{P};\hat{\mathfrak{g}})\simeq \widehat{\operatorname{Conn}}(\hat{P};\hat{\mathfrak{g}})$ and $C\equiv 0$ we say that the extension is \emph{fully complete}.
$$
\xymatrix{0 \ar[r] & \operatorname{Conn}(\hat{P};\hat{\mathfrak{g}}) \ar@{^(->}[r] & \widehat{\operatorname{Conn}}(\hat{P};\hat{\mathfrak{g}})  \ar@{^(->}[r] & \Omega ^1_{eq}(\hat{P};\hat{\mathfrak{g}}) \\
& \Omega ^1_{eq}(P;\mathfrak{g}) &  \ar@{_(->}[l] \operatorname{Conn}(P;\mathfrak{g}) & \ar[l]^-{\delta} \ar@{_(->}[lu]_{\jmath} C^1(\hat{P};\hat{\mathfrak{g}}) & \ar[l] 0  }
$$

Physically, a complete extension $\hat{S}^{\hat{G}}$ of a YMT theory $S^G$ is one which reproduces $S^G$ \textit{exactly} (i.e, without any correction term) when restricted to the subspace $C^1(\hat{P};\hat{\mathfrak{g}})$. Furthermore, it is a genuine extension of $S^G$ in the mathematical sense. If we think of this subspace as determining the configurations of $\hat{S}^{\hat{G}}$ in a certain scale, then we conclude that a complete extension of $S^G$ is one which reproduces $S^G$ exactly in \textit{some scales}. In turn, a fully complete extension is one describing $S^G$ exactly in \textit{every scales}. This suggests a relation between emergence phenomena and complete extensions, which will be confirmed in Subsection  \ref{sec_emergence}.

Finally, a full extension of $S^G$ is one whose vector potentials are not connection 1-forms on $\hat{P}$, but actually arbitrary equivariant 1-forms. Thus, in them, we have a physical meaning for 1-forms $\hat{D}:T\hat{P}\rightarrow \hat{g}$ which are not vertical. 
\begin{remark}[linear and equivariant extensions]\label{remar_linear_equivariant_extensions}
The definition of extension above has some variations. Indeed, as defined above, in an extension of a YMT theory the extended domain is required to be a $\operatorname{Gau}_{\hat{G}}(\hat{P})$-set such that the extended functional is $\operatorname{Gau}_{\hat{G}}(\hat{P})$-invariant. No structures are required on the correction subspace and the correction functional and the $\delta$-map are just functions. The situation can be improved in two directions: 
\begin{enumerate}
    \item by requiring that the extended domain and the correction subspace are actually linear subspaces of $\Omega ^1_{eq}(\hat{P};\hat{\mathfrak{g}})$. We can also require linearity of the correction function and/or of the $\delta$-map;
    \item by requiring that not only the extended domain, but also the correction subspace is a $\operatorname{Gau}_{\hat{G}}(\hat{P})$-set. In this situation it is typically useful to require equivariance of $C$ and $\delta$, so that the whole structure is equivariant. 
\end{enumerate}
In the first case, we say that the extension is \textit{partially linear} (resp. \textit{linear}) if the extended domain and the correction subspaces have linear structures (resp. if in addition $C$ and $\delta$ are linear). Notice that since the action of $\operatorname{Gau}_{\hat{G}}(\hat{P})$ preserves the linear structure of $\Omega^{1}_{eq}(\hat{P};\hat{\mathfrak{g}})$, it follow that it also preserves the linear structure of the vector subspace $\operatorname{Conn}(\hat{P};\hat{\mathfrak{g}})$. In the second case, we say that the extension is \textit{weakly equivariant} if the correction subspace is $\operatorname{Gau}_{\hat{G}}(\hat{P})$-invariant and \textit{equivariant} if in addition $C$ and $\delta$ are equivariant. Notice that if $\delta \neq id$, then we may have equivariant extensions of YMT which are not gauge invariant. The linear context is useful when studying complete extensions \cite{complete_extension}, while the equivariant context in useful in the study of incomplete extensions \cite{gauge_breaking_YMT_extensions}. 

\end{remark}
\subsection{Trivial Extensions}

\quad\;\, Each YMT theory admits some trivial extensions, as described in the following examples.

\begin{example}[\textit{null-type extensions}]\emph{\label{null_extension}
Every YMT $S^G$ admits an  extension relative to any basic extension $\imath:G \hookrightarrow \hat{G}$ and whose space is any given $\operatorname{Gau}_{\hat{G}}(\hat{P})$-invariant subset $\widehat{\operatorname{Conn}}(\hat{P};\hat{\mathfrak{g}})$ satisfying (\ref{domain_extension}). It is obtained by taking the null extended functional $\hat{S}^{\hat{G}}=0$, null correction subspace, i.e, $C^1(\hat{P};\hat{\mathfrak{g}})=0$, null  correction functional $C=0$ and $\delta =0$. Condition  \ref{decomposition_extension} is immediately satisfied. We say that this is the \emph{null extension of $S^G$ with extended space $\widehat{\operatorname{Conn}}(\hat{P};\hat{\mathfrak{g}})$}. In the case when the extended space is the space of connections $\operatorname{Conn}(\hat{P};\hat{\mathfrak{g}})$, we say simply that it is the \emph{null extension of $S^G$}. This is an example of an extension which is equivariant even if $S^G$ is not gauge invariant.}
\end{example}

\begin{example}[identity extensions]\label{identity_extension}\emph{
For every gauge invariant\footnote{Here is an example where the hypothesis of gauge invariance is needed.} YMT $S^G$ we can assign the \textit{identity extension}, relative to the identity basic extension $id:G\hookrightarrow G$, defined as follows. The extended domain and the correction subspace are 
\begin{equation}\label{identity_extension_domain}
    \widehat{\operatorname{Conn}}(P;\mathfrak{g})=\operatorname{Conn}(P;\mathfrak{g})=C^1(P;\mathfrak{g}).
\end{equation}
The $\delta$-map is the identity of $\operatorname{Conn}(P;\mathfrak{g})$ and $C\equiv 0$. We clearly have $S^G\circ \delta +C=S^G$, so that this realizes $S^G$ as a complete extension of itself.}
\end{example}

\begin{example}[equivariant choice extension]\emph{As in the last example, consider the identity basic extension $\operatorname{id}:G\hookrightarrow G$, with same correction subspace $C^1(P;\mathfrak{g})=\operatorname{Conn}(P;\mathfrak{g})$, with same correction functional $C\equiv 0$ and with same $\delta$-map, i.e, $\delta = \operatorname{id}$. But instead of taking the extended domain as in (\ref{identity_extension_domain}), take $\widehat{\operatorname{Conn}}(P;\mathfrak{g})$ as some $\operatorname{Gau}_G(P)$-invariant subset of $\Omega ^1(P;\mathfrak{g})$ containing $\operatorname{Conn}(P;\mathfrak{g})$, i.e, which satisfies (\ref{domain_extension}). Then, by the equivariant version of the Axiom of Choice\footnote{More precisely, if we assume the Axiom of Choice, then for every group $G$, in the corresponding category of $G$-sets, every mono and every epi are split \cite{category_G_sets}.}, it follows that there are $\operatorname{Gau}_G(P)$-equivariant retracts, as below. If $r$ is one such retract, define $\hat{S}^G=S^G\circ r$. Then $\hat{S}^G\vert_{C^1(P;\mathfrak{g})}=S^G = S^G\circ \delta +C$, which means that  $\hat{S}^G$ is another extension of $S^G$. But notice that it is just the identity extension with a larger extended domain arising from the Axiom of Choice.
$$
\xymatrix{
\operatorname{Conn}(P;\mathfrak{g}) \ar@/_{0.5cm}/[rr]_{\operatorname{id}} \ar@{^(->}[r]& \ar[r]^-{r} \widehat{\operatorname{Conn}}(P;\mathfrak{g})  &  \operatorname{Conn}(P;\mathfrak{g})}
$$
} 
\end{example}
\begin{remark}
A look at the last example reveals that one can actually extend the domain of every equivariant extension such that $\delta=\operatorname{id}$. In particular, we can take $\widehat{\operatorname{Conn}}(P;\mathfrak{g})=\Omega^1_{eq}(P;\mathfrak{g})$ in order to get a full extension. For a detailed discussion, see \cite{gauge_breaking_YMT_extensions}. 
\end{remark}

\begin{example}[constant extensions]\emph{ Notice that if $c\in \mathbb{R}$ is any real number and $X$ is any $G$-set, then the constant function $f$ with value $c$ is $G$-invariant, because $f(g\cdot x)=c=f(x)$. With this in mind, fixed a basic extension $\imath:G\hookrightarrow \hat{G}$, let $\widehat{\operatorname{Conn}}(\hat{P};\hat{\mathfrak{g}})$ be any $\operatorname{Gau}_{\hat{G}}(\hat{P})$-invariant set satisfying (\ref{domain_extension}) and let $C^1(\hat{P};\hat{\mathfrak{g}})$ be any subset of it. Given a real number $c\in \mathbb{R}$ and a connection $D_0$ in $P$ we will build an extension for every YMT theory $S^G$ whose extended domain and whose correction subspace are the above. Define $\delta :C^1(\hat{P};\hat{\mathfrak{g}})\rightarrow \operatorname{Conn}(P;\mathfrak{g}$) as the constant map in $D_0$. Define $C:C^1(\hat{P};\hat{\mathfrak{g}})\rightarrow \mathbb{R}$ as the constant map with value $c$. Define $\hat{S}^{\hat{G}}:\widehat{\operatorname{Conn}}(\hat{P};\hat{\mathfrak{g}}) \rightarrow \mathbb{R}$ as the constant map in $S^G[D_0]+c$, which is $\operatorname{Gau}_{\hat{G}}(\hat{P})$-invariant by the previous remark. Furthermore, we clearly have $\hat{S}^{\hat{G}}\circ \jmath=S^G\circ \delta + C$, so that the described data really defines an extension of $S^G$.} 
\end{example}

An extension which is not of the types above is called \textit{nontrivial}. Of course, we will be more interested in these ones. There are a lot of well-known examples of nontrivial extensions for specific kinds of YMT theories. In the following we will describe some of them. We begin with a detailed example.

\subsection{A Detailed Example: Higgs Mechanism}\label{sec_higgs}

\quad\;\,It is well known that the choice of a specific scalar vacuum field with non zero expectation value coupling with the YM gauge field causes a spontaneous symmetry breaking (SSB) of the classical YM theory  \cite{HIGGS1964132,PhysRevLett.13.508,PhysRevLett.13.321}. Here we will see that any equivariant correction of a YMT $S^G$ theory involving the classical Higgs fields (in particular the corresponding Yang-Mills-Higgs (YMH) theory) can be viewed as an extension of $S^G$ in the sense of the previous section. 

Recall that for a principal $\hat{G}$-bundle $\hat{P}$ and for a closed Lie subgroup $G\subset \hat{G}$, the space of \emph{classical Higgs fields}, denoted by $\operatorname{Higgs}(\hat{P};G)$, is the space $\Gamma (\hat{P}/G)$ of global sections $\phi$ of the quotient bundle $\hat{P}/G$, which can be identified with the associated bundle $\hat{P}\times_{\hat{G}} (\hat{G}/G)$ \cite{sardanashvily2016classical,v.Westenholz1980}. But sections of this associated bundle are in 1-1 correspodence with $G$-reductions of $\hat{P}$ \cite{kobayashi1996foundations}, so that classical Higgs fields are in bijection with these $G$-reductions. In this sense, for a classical field theory $S[s]$ depending on $\hat{P}$, we say that \emph{the $\hat{G}$-symmetry  is spontanously broken to a $G$-symmetry} when such a classical Higgs field exists. In this case, the action functional $S$ couples with the Higgs fields by means of adding a term  $C[s,\phi]$ depending on both the fields, so that we have a total action $S_C=S+C$. Thus, in the case of a YM theory $S^G$ on $P$, we have $S^G_C[D,\phi]=S^G[D]+C[D,\phi]$, suggesting that one can regard  the functionals $S^G_C$ as incomplete extensions of $S^G$. 

Let $P$ be a principal $G$-bundle and let $\imath:G\hookrightarrow \hat{G}$ be a basic extension of $G$. The space $\operatorname{Higgs}(\hat{P};G)$ is clearly nonempty, since $\hat{P}$ arises from $P$ by the basic extension. We say that $\imath:G\hookrightarrow \hat{G}$ is \textit{coherent} if it becomes endowed with a map $\theta:\hat{G}/G\rightarrow \mathfrak{\hat{g}}$ satisfying the compatibility relation $\theta(aG) = Ad_{a}(\theta(G))$, where $a\in \hat{G}$. In this case, we have a bundle morphism $\Theta:\hat{P}\times_{\hat{G}} (\hat{G}/G)\rightarrow E_{\hat{\mathfrak{g}}}$, which is fiberwise given by $\theta$. Consequently, we have also have an induced map $\Theta_*:\operatorname{Higgs}(\hat{P};G)\rightarrow \Gamma (E_{\hat{\mathfrak{g}}})$. Now, notice that given a connection $\hat{D}$ in $\hat{P}$, its covariant derivative $d_{\hat{D}}:\Omega^0(\hat{P};\hat{\mathfrak{g}})\rightarrow \Omega^1(\hat{P};\hat{\mathfrak{g}})$ induces a covariant derivative $\nabla_{\hat{D}}:\Gamma (E_{\hat{\mathfrak{g}}})\rightarrow \Omega^1(M; E_{\hat{\mathfrak{g}}})$ on the adjoint bundle. We can then take the image of the following composition: 
$$
\xymatrix{\ar@/_{0.5cm}/[rrrr]_{\Theta_{*,\hat{D}}} \operatorname{Higgs}(\hat{P};G) \ar[r]^-{\Theta_{*}} & \Gamma (E_{\hat{\mathfrak{g}}}) \ar[r]^-{\nabla_{\hat{D}}} & \Omega^1(M; E_{\hat{\mathfrak{g}}}) \ar[r]^-{\simeq}  & \Omega^1_{heq}(\hat{P};\hat{\mathfrak{g}}) \ar@{^(->}[r] & \Omega^1_{eq}(\hat{P};\hat{\mathfrak{g}}). 
}
$$

In particular, one can take the union of $\operatorname{img}(\Theta_{*,\hat{D}})$ over all $\hat{D}\in \operatorname{Conn}(\hat{P};\hat{\mathfrak{g}})$, obtaining a subset $\Theta^1_{\operatorname{conn}}(\hat{P};\hat{\mathfrak{g}})\subset \Omega^1_{eq}(\hat{P};\hat{\mathfrak{g}})$ independent of the choice of $\hat{D}$. As one can check, since $\Theta_{*}$ is equivariant, this subset is invariant under the canonical action by pullbacks of $\operatorname{Gau}_{\hat{G}}(\hat{P})$. Thus, the union $\Theta^1_{\operatorname{conn}}(\hat{P};\hat{\mathfrak{g}})\cup \operatorname{Conn}(\hat{P};\hat{\mathfrak{g}})$ is also $\operatorname{Gau}_{\hat{G}}(\hat{P})$-invariant and clearly satisfies (\ref{domain_extension}). Notice that this union is disjoint except maybe for the null 1-form. This means that the first inclusion in the diagram below is well defined, so that we can consider the corresponding composition. 

\begin{equation}\label{delta_higgs}
    \xymatrix{\ar@/_{0.5cm}/[rr]_{\delta} \Theta^1_{\operatorname{conn}}(\hat{P};\hat{\mathfrak{g}}) \ar@{^(->}[r] & \Theta^1(\hat{P};\hat{\mathfrak{g}})\times \operatorname{Conn}(\hat{P};\hat{\mathfrak{g}}) \ar[r]^-{\operatorname{pr}_2 } & \operatorname{Conn}(\hat{P};\hat{\mathfrak{g}})}
\end{equation}

Given any $\operatorname{Gau}_{\hat{G}}(\hat{P})$-invariant functional $C:\Theta^1_{\operatorname{conn}}(\hat{P};\hat{\mathfrak{g}})\rightarrow \mathbb{R}$ and any gauge invariant YMT theory $S^{\hat{G}}$, we see that the construction above realizes the sum $\hat{S}^{\hat{G}}_C=S^{\hat{G}}\circ \delta + C$ as an extension of the  YMT theory  $S^{\hat{G}}$, relative to the trivial basic extension $id:\hat{G}\hookrightarrow \hat{G}$, with extended domain and correction subspaces given by 
$$
\widehat{\operatorname{Conn}}(\hat{P};\hat{\mathfrak{g}})=\Theta^1_{\operatorname{conn}}(\hat{P};\hat{\mathfrak{g}})=C^1(\hat{P};\hat{\mathfrak{g}}),
$$ 
whose $\delta$-map is that given in (\ref{delta_higgs}) and whose correction functional is the one given $C$. In particular, if $S^{\hat{G}}$ is a classical Yang-Mills theory with pairing $\langle \cdot, \cdot \rangle_{\hat{\mathfrak{g}}}$ and $$C(\hat{D},\phi)=\langle \Theta_{*,\hat{D}}\phi, \Theta_{*,\hat{D}}\phi \rangle +V(\phi),$$ where $V(\phi)$ is some potential, then  $\hat{S}^{\hat{G}}_C$ is precisely the action functional of the corresponding Yang-Mills-Higgs theory. Therefore, \textit{the Yang-Mills-Higgs theory is an example of an extension of YM theories in the sense of previous section}. Notice that in this case, we actually obtain an example of equivariant extension. 

Another practical example of this construction is the following. 

\begin{example}[t'Hooft-Polyakov monopole]\emph{
Recall the identification $SO(3)/SO(2)\simeq \mathbb{S}^2$, so that  there is an injective map $SO(3)/SO(2)\rightarrow \mathbb{R}^3$. Under the Lie algebra identification   $\mathbb{R}^3\simeq \mathfrak{so}(3)$ one can check that $\theta(aSO(2))= Ad_{a}(\theta(SO(2)))$, so that the basic extension $SO(2)\hookrightarrow SO(3)$ is coherent and we have the bundle morphism $\Theta:{\hat{P}}\times_{SO(3)}(SO(3)/SO(2))\rightarrow E_{\mathfrak{so}(3)}$. Fixing a local trivialization for $\hat{P}$ and parameterizing $SO(3)$ by the Euler angles $\varphi, \alpha,\psi$, we see that the map $\Theta _*:\Gamma(\operatorname{Higgs}(\hat{P};SO(2))\rightarrow \Gamma (E_{\mathfrak{so(3)}})$, via the fiberwise identification of $\theta$, takes the following form, which is the classical expression of the normalized Higgs field in the t'Hooft-Polyakov monopole description \cite{v.Westenholz1980}: 
\begin{equation*}
    \Theta_{*}(\phi) = \begin{pmatrix}
    0 & -\operatorname{cos}(\alpha) & -\operatorname{cos}(\varphi)\operatorname{sin}(\alpha)\\
    \operatorname{cos}(\alpha) & 0 & -\operatorname{sin}(\varphi)\operatorname{sin}(\alpha) \\
    \operatorname{cos}(\varphi)\operatorname{sin}(\alpha) & \operatorname{sin}(\varphi)\operatorname{sin}(\alpha) & 0 
    \end{pmatrix}\simeq \begin{pmatrix}\operatorname{sin}(\varphi)\operatorname{sin}(\alpha) \\
    -\operatorname{cos}(\varphi)\operatorname{sin}(\alpha) \\
    \operatorname{cos}(\alpha)
\end{pmatrix}.
\end{equation*}}
\end{example}
\subsection{Further Examples}\label{sec_further_examples}

\quad\;\, Some additional manifestations of extensions of YMT theories are the following. It is not our aim to exhaust the list of all known examples, but only to show the wide range of applicability of the definitions in Section \ref{sec_extensions}.

\begin{example}[BF Theory]\label{example_BF}
\emph{Let $P$ be a principal $G$-bundle on a compact oriented 4-dimensional manifold $(M,\omega)$ and let $\langle\cdot,\cdot\rangle$ be a pairing in horizontal equivariant 2-forms. In analogy to the classical BF theories \cite{Baez:1999sr,derek}, define the \textit{BF-type theory} with the given pairing as the functional
\begin{equation}\label{bfaction}
  S_{BF}:\operatorname{Conn}(P,\mathfrak{g})\times \Omega^{2}_{he}(P,\mathfrak{g})\rightarrow \mathbb{R} \quad \text{such that}\quad   S_{BF}(D,B) = \int_{M}\langle F_{D},B\rangle \omega.
\end{equation}
Now, consider the identity basic extension $id:G\hookrightarrow G$ and for $n=4$ look at the map $\operatorname{curv}:\operatorname{Conn}(P;\mathfrak{g})\rightarrow \Omega ^2_{heq}(P;\mathfrak{g})$ assigning to each connection $D$ in $P$ its curvature 2-form $F_D$. Let $\operatorname{grph}(\operatorname{curv})$ be its graph. The projection $\operatorname{pr}:\operatorname{grph}(\operatorname{curv})\rightarrow \operatorname{Conn}(P;\mathfrak{g})$ is a bijection whose inverse is the diagonal map $\Delta_{\operatorname{curv}}$. Notice that $S_{BF}\circ \Delta_{\operatorname{curv}}=S^G$, where $S^G$ is a gauge invariant YMT theory with the given pairing. Thus, with the same extended domain, same correction subspace, same $\delta$-map and same correction functional, we see that  $S_{BF}\circ \Delta_{\operatorname{curv}}$ is an extension of $S^G$ which in some sense is \emph{equivalent} to the identity extension of Example \ref{identity_extension}. In Section \ref{sec_category_extensions} we will see that they are actually isomorphic as objects of the category of extensions.}
\end{example}
 

\begin{example}[Higgs again]\emph{
Let us revisit our discussion of Subsection \ref{sec_higgs}. For the present discussion, we will look at Higgs fields $\phi\in \operatorname{Higgs}(\hat{P};G)$ which are \textit{$\hat{D}$-Higgs vacuum}, i.e., which are parallel with respect to some exterior covariant derivative $\nabla_{\hat{D}}$, i.e, $\nabla_{\hat{D}}\phi = 0$. In our context we note that the existence of a Higgs field is guaranteed by the definition of the extension, once that $P$ is realized as a $G$-reduction of $\hat{P}$ by the monomorphism $\imath: P\rightarrow \hat{P}$, but, a priori, those \textit{$\hat{D}$-Higgs vacuum} need not exist. Indeed, a $\hat{G}$-connection $\hat{D}$ in $\hat{P}$ reduces to a $G$-connection in $P$ iff $P$ is obtained from a $\hat{D}$-parallel Higgs field vacuum \cite{kobayashi1996foundations}. Thus, let $\operatorname{Higgs}_0(\hat{P};G)\subset \operatorname{Higgs}(\hat{P};G)$ be the set of Higgs vacuum. For a fixed $\phi_0\in \operatorname{Higgs}_0(\hat{P};G)$, define $C^{1}(\hat{P},\hat{\mathfrak{g}}) = \{\hat{D}\in \operatorname{Conn}(\hat{P},\hat{\mathfrak{g}})| \nabla_{\hat{D}}\phi_{0} = 0\}$. By the equivariance of the exterior covariant derivative it is clear that $C^{1}(\hat{P},\hat{\mathfrak{g}})$ is $\operatorname{Gau}_{\hat{G}}(\hat{P})$-invariant. By construction, there exists the map $\delta: C^{1}(\hat{P},\hat{\mathfrak{g}})\rightarrow \operatorname{Conn}(P;\mathfrak{g})$ which takes a connection in which $\phi_0$ is parallel and gives it reduction to $P$. Let $S^G$ be a gauge invariant YMT theory whose instanton sector is $P$. Then, for every $\operatorname{Gau}_{\hat{G}}(\hat{P})$-invariant set $\widehat{\operatorname{Conn}}(\hat{P};\hat{\mathfrak{g}})$  satisfying (\ref{domain_extension}), and every $\operatorname{Gau}_{\hat{G}}(\hat{P})$-invariant functional $C:C^{1}(\hat{P},\hat{\mathfrak{g}})\rightarrow \mathbb{R}$, the sum $\hat{S}^{\hat{G}}=S^G\circ \delta + C$ is an extension of $S^G$.
}
\end{example}

\begin{remark}
Note the difference: while in Subsection \ref{sec_higgs} the Yang-Mills-Higgs theory with group $\hat{G}$, relatively to the basic extension $\imath:G\hookrightarrow \hat{G}$, was realized as an extension of YMT theory with the \textit{same} group $\hat{G}$, in example above the Yang-Mills-Higgs with group $\hat{G}$ was realized as an extension of the YMT with \textit{smaller} group $G$. This naturally leads one to ask whether a YMT $S^{\hat{G}}$ can be itself regard as an extension of $S^G$. This was considered and studied in \cite{gauge_breaking_YMT_extensions}.
\end{remark}
\begin{remark}
Higgs-like mechanisms were observed to emerge as extensions of YMT theories in some other contexts. For instance, by the method of dimension reduction if the extension is allowed to lie in a higher dimensional space-time \cite{forgacs1980,Harnad} and if the Lie algebra of the YM theory is replaced by a Leibniz algebra giving the $2$-Higgs mechanism introduced in \cite{PhysRevD.99.115026}.
\end{remark}

\begin{example}[Full YM theories]\emph{ Let $P$ be a $G$-principal, $(M,g)$ a semi-Riemannian manifold and let $S^G$ the classical YM theory, whose pairing is given by $\langle \alpha,\beta \rangle_{\mathfrak{g}}=\operatorname{tr}(\alpha [\wedge]\star_g\beta)$. In dimension $n=4$ we have another pairing given by  $\langle \alpha,\beta \rangle'_{\mathfrak{g}}=\operatorname{tr}(\alpha [\wedge]\beta)$. In the literature, the action functional 
\begin{align}
    S_{\operatorname{full}}^G:\operatorname{Conn}(P;\mathfrak{g})\rightarrow \mathbb{R}\quad \text{given by}\quad S_{\operatorname{full}}^G[D]&=\int_M\langle F_D, F_D \rangle_{\mathfrak{g}}dvol_g + \int_M\langle F_D, F_D \rangle_{\mathfrak{g}}'dvol_g\\
     &= S^G[D]+S_{\operatorname{top}}^G[D]
\end{align}
is referred as the \textit{full YM theory} with gauge group $G$. Let us consider the identity basic extension $id:G\hookrightarrow G$ with extended domain and correction subspace as in (\ref{identity_extension_domain}). Thus, taking $\delta=id$ we see that the full YM theory is an extension of the YM theory with $C=S_{\operatorname{top}}$.}
\end{example}

\begin{example}[topological extensions]\emph{
The functional $S^G_{\operatorname{top}}$ of the last example is topological, since by Chern-Weil homomorphism the form $\operatorname{tr}(F_D[\wedge]F_D)$ is closed and therefore is realized in de Rham cohomology. More generally, supposing $n=2k$, let $\kappa\in \operatorname{inv}(\mathfrak{g})$ be a homogeneous invariant polynomial of degree $k$ in $\mathfrak{g}$ and define $S^G_{\kappa}[D]=\int_M\kappa(F_D[\wedge]\cdots [\wedge]F_D)$. Then there is an extension of the classical YM theory whose correction functional is $C=S^G_{\kappa}$.}
\end{example}
\begin{example}[background field]\emph{
Let $S^G$ be a gauge invariant YMT theory and let $E\rightarrow M$ be some field bundle. Let $C:\operatorname{Conn}(P;\mathfrak{g})\times \Gamma (E)\rightarrow \mathbb{R}$ be any $\operatorname{Gau}_G(P)$-invariant functional, where the gauge groups acts trivially on $\Gamma(E)$. Then, for every fixed $s\in\Gamma(E)$, the sum $S^G+C_s$, where $C_s[D]=C[D,s]$, is an extension of $S^G$ relatively to the identity basic extension $id:G\hookrightarrow G$. Thus, in particular, \textit{every background theory minimally coupled with a gauge invariant YMT theory induces an extension.}} 
\end{example}

In our definition of extension for a YMT, we constrained the extended functional $\hat{S}^{\hat{G}}$ to have domain satisfying (\ref{domain_extension}). Since the correction subspace is required to be a subset of the extended domain, the upper bound in (\ref{domain_extension}) also applies to $C^1(\hat{P};\hat{\mathfrak{g}})$. We notice that if we avoid this upper bound in both the extended domain and correction subspace, then a lot of new situations can be regarded as new examples. Let us call them ``partial examples'' of YMT extensions.

\begin{example}[non-commutative YM via Seiberg-Witten maps]\emph{In \cite{SW_map} it was argued the remarkable existence of a map assigning a noncommutative analogue $S_{\operatorname{nc}}$ to every classical gauge theory $S$. Furthermore, this noncommutative theory is supposed to be expanded in a noncommutative parameter $\theta$, i.e, $S_{\operatorname{nc}}=\sum_{i\geq 0} \theta^i S_i$, such that $S_0=S$. Let $\operatorname{NConn}(P;\mathfrak{g})$ be the space of those ``noncommutative connections'' in which $S_{\operatorname{nc}}$ is defined and consider the map $\delta:\operatorname{NConn}(P;\mathfrak{g})\rightarrow \operatorname{Conn}(P;\mathfrak{g})$ given by the ``commutative limit'' $\theta \rightarrow 0$. Thus, in the case of a YMT theory one could write $S_{\operatorname{nc}}=S^G\circ \delta + C$, where $C=\sum_{i\geq 1}\theta^iS_i$, leading us to ask if $S_{\operatorname{nc}}$ can be considered as an extension of $S^G$. Although the space $\operatorname{Conn}(P;\mathfrak{g})$ was not rigorously defined, there is a more fundamental problem in regarding $S_{\operatorname{nc}}$ as an extension. Notice that a priori one can take each $S^G\circ \delta+\theta^i S_i$ as an extension of $S^G$. Since as in Subsection \ref{sec_monoid} we will prove that the sum of extensions remains an extension, we conclude that $\hat{S}^G_{k}=S^G\circ \delta + S_{\operatorname{nc},k}$, where $S_{\operatorname{nc},k}=\sum_{1\leq i \leq k} \theta^iS_i$, is an extension of $S^G$ for every $k$. The problem may occur in the limiting process $k\rightarrow \infty$, since in this case we could analyze the convergence of a series in the space of all extensions of $S^G$. In Subsection \ref{sec_bundle} we prove that this space has a natural topology when $P$ is compact, so that this convergence could be really considered. This, however, is outside of the scope of this work, so that we only present an speculation:
\begin{itemize}
    \item \textbf{Speculation}. \emph{Suppose $P$ compact.  Then  $\hat{S}^G_{k}$ is an extension of $S^G$ for every $k>0$. Furthermore, the limit $\lim _{k\rightarrow \infty}\hat{S}^G_{k}$ exists in the topology of Subsection \ref{sec_bundle} and is $S^G\circ\delta +S_{\operatorname{nc}}$.}  
\end{itemize}}
\end{example}

\begin{example}[deformations]
\emph{Analogous discussions and speculations of the last example apply to other ways to deform the action functional of  YMT theory via auxiliary parameters, such as those in \cite{YM_deformation_1,YM_deformation_2,YM_deformation_CS,rivelles}}.
\end{example}

\begin{example}[tensorial YM theories]\emph{In the sequence of papers \cite{tensor_gauge_1,tensor_gauge_2,tensor_gauge_3} it was introduced and studied what the authors called \textit{non-abelian tensor gauge theories}. These are some version of higher gauge theory such that vector potentials of arbitrary degree, such as $A_{i_1,i_2,...}$, are allowed. In particular, its Lagrangian density is suppose to be the sum $\mathcal{L}=\sum_{i\geq 1}g_i\mathcal{L}_i$, where $g_1=1$ and $\mathcal{L}_1$ is the Lagrangian of the classical YM theory. Furthermore, $\mathcal{L}_i$, with $i>1$ depends on a gauge field of higher degree. Let $\operatorname{Tens}(P;\mathfrak{g})$ be the space of all those ``higher gauge fields'' and consider the map $\delta:\operatorname{Tens}(P;\mathfrak{g}) \rightarrow \operatorname{Conn}(P;\mathfrak{g})$ which projects onto the space of gauge-fields of degree one. Then $\int\mathcal{L}=S^G\circ \delta + \sum_{i\geq 2} \int \mathcal{L}_i$ and one could ask, in analogy to the previous examples, whether this realizes $\mathcal{L}$ as an extension of $S^G$. Here we have an additional problem: $\delta:\operatorname{Tens}(P;\mathfrak{g})$ is too large to be contained in $\Omega^1_{eq}(P;\mathfrak{g})$. Thus, this cannot be an extension domain. But, even if we admit larger extended domains, the results of Subsection \ref{monoid_structure}  and of Subsection \ref{bundle} cannot be applied here. }
\end{example}

In trying to avoiding misunderstandings, we close this section giving some non-examples in which the expression ``Yang-Mills extensions'' is used in a sense which is (as far as the authors know) completely unrelated with the notion introduced here: the well-known notion of \textit{supersymmetric extensions}, the \textit{stringy extensions} of \cite{YM_extension_stringy}, the \textit{Yang-Mills families} of \cite{YM_extension_families} and the extensions studied in \cite{YM_extension_book}.  

\subsection{The Role of Emergence Phenomena} \label{sec_emergence}

\quad\;\,Emergence phenomena arise in the most different sciences: biology, arts, philosophy, physics, and so on. In each of these contexts the term ``emergence'' has a different meaning \cite{emergence_0,emergence_1,emergence_2}, but all are about observing characteristics of a system in another system in an unexpected way. In \cite{emergence_yuri} the authors suggested an axiomatization for the notion of emergence phenomena between field theories. Here we will show the surprising fact that in the theories emerging from a YMT theory the $S^G$ itself can be regarded as a new nontrivial example of an extension of $S^G$.

Following \cite{emergence_yuri} we begin by recalling that a \textit{parameterized field theory} over a smooth manifold $M$ consists of bundles $E\rightarrow M$ (the field bundle) and $P\rightarrow M$ (the \textit{parameter bundle}), a set of \textit{parameters} $\operatorname{Par}(F)\subset \Gamma(F)$, a set of \textit{field configurations} $\operatorname{Conf}(E)$ and a family of functionals $S_{\varepsilon}:\operatorname{Conf}(E)\rightarrow \mathbb{R}$, with $\varepsilon\in \operatorname{Par}(F)$. The most interesting cases occur when dependence of $S_{\varepsilon}$ on $\varepsilon$ is only on certain differential operators. More precisely, we say that a parameterized theory is of \textit{differential type} if there is a 
map $D_{-}:\operatorname{Par}(F)\rightarrow \operatorname{Diff}(E)$ assigning to each parameter $\varepsilon$ a differential operator $D_{\varepsilon}: \Gamma(E)\rightarrow \Gamma(E)$ in $E$ and a pairing  $\langle \cdot , \cdot \rangle$ in $\Gamma(E)$ such that
  $S_{\varepsilon}(s)=\int\langle s,D_{\varepsilon}s \rangle$.   In this case, $\operatorname{Conf}(E)=\Gamma (E)$.

Let $S^G$ be a YMT. We will say that it has \text{perfect pairing} if the underlying pairing defining it is a perfect pairing. For instance, every Yang-Mills theory has perfect pairing.

\begin{lemma}\label{lemma_YMT_parameterized}
Under the choice of a connection $D_0$ in the instanton sector $P$, every YMT theory $S^G$ with perfect pairing can be naturally regarded as a parameterized field theory of differential type for any given set of parameters $\operatorname{Par}(F)$.
\end{lemma}
\begin{proof}
Recalling that $ \operatorname{Conn}(P;\mathfrak{g})$ is an affine space of $\Omega ^1(M;E_\mathfrak{g})$, the choice of $D_0$ allows us to regard the action functional $S^G$ as defined in $\Omega ^1(M;E_\mathfrak{g})\simeq \Gamma (TM^*\otimes E_{\mathfrak{g}})$. Furthermore, since the underlying pairing $\langle \cdot , \cdot \rangle $ is perfect, $d_D$ has an adjoint which we denote by $d^{*}_D$, so that $$
S^G(D)=\int\langle d_DD, d_DD \rangle =\int\langle D, d^{*}_D(d_DD) \rangle. 
$$
For every $\operatorname{Par}(F)$ take the constant function $\operatorname{cst}:\operatorname{Par}(F)\rightarrow \operatorname{Diff}(TM^*\otimes E_{\mathfrak{g}})$ given by $\operatorname{cst}(\varepsilon)=d^{*}_D\circ d_D$. Thus $S_{\varepsilon}(D)=S^G(D)$ for every $\varepsilon$, which  clearly  regards $S^G$ as a paremeterized theory of differential type with set of parameters $\operatorname{Par}(F)$.
\end{proof}

We say that a parameterized theory $S_{1}$, with field bundle $E_1$ and set of parameters $\operatorname{Par}(F_1)$,  \textit{strongly emerges} from another parameterized theory $S_2$, with field bundle $E_2$ and set of parameters $\operatorname{Par}(F_2)$, if there are maps $F:\operatorname{Par}(F_1)\rightarrow \operatorname{Par}(F_2)$ and $G: \operatorname{Conf}(E_1)\rightarrow \operatorname{Conf}(E_2)$ such that $S_{1,\varepsilon}(s)=S_{F(\varepsilon)}(G(s))$ for every $\varepsilon$ and $s$. The pair $(F,G)$ is called a \textit{strong emergence phenomenon} between $S_1$ and $S_2$.

\begin{theorem}

Let $S^G$ be a YMT theory with perfect pairing, let $F_1$ be a bundle and let $\operatorname{Par}(F_1)\subset \Gamma (F_1)$ be a subset. Regard $S^G$ as a parameterized theory with set of parameters $\operatorname{Par}(F_1)$, as in Lemma \ref{lemma_YMT_parameterized}. Suppose that $S^G$ emerges from another parameterized theory $S_2$, with 
$$
\operatorname{Conn}(\hat{P};\hat{\mathfrak{g}})\subseteq \operatorname{Conf}(E_2)\subseteq \Omega^1_{eq}(\hat{P};\hat{\mathfrak{g}}) 
$$ 
being $ \operatorname{Gau}_{\hat{G}}(\hat{P})$-equivariant. Then for every emergence phenomena $(F,G)$ between $S_2$ and $S^G$, for every $\varepsilon \in \operatorname{Par}(F_1)$ the action functional $S_{2,\varepsilon}$ is a complete extension of $S^G$.
\end{theorem}
\begin{proof}
From Lemma \ref{lemma_YMT_parameterized} $S^G_{\delta}=S^G$ for every $\delta$, while from the definition of emergence phenomena we have $S_{2,\varepsilon}(s)=S^G{2,F(\varepsilon)}(G(s))$. Thus, for every $\varepsilon$ we have $S_{2,\varepsilon}=S^G\circ G$. Therefore, $S_{2,\varepsilon}$ is an extension of $S^G$ with $\widehat{\operatorname{Conn}}(\hat{P};\hat{\mathfrak{g}})=\operatorname{Conf}(E_2)=C^1(\hat{P};\hat{\mathfrak{g}})$ with $\delta = G$ and $C=0$.
\end{proof}

\section{Space of Extensions}\label{sec_space_extensions}
\quad\;\,Let $S^G$ be a YMT and, for a fixed basic extension $\imath:G\hookrightarrow \hat{G}$ of $G$, let $\operatorname{Ext}(S^G;\hat{G})$ denote the space of all extensions of $S^G$ whose underlying basic extension is $\imath$. Here we study the properties of this space. In Subsection \ref{sec_monoid} we prove that for each fixed extended domain $\widehat{\operatorname{Conn}}(\hat{P};\hat{\mathfrak{g}})$ the corresponding subset $\operatorname{Ext}(S^G;\widehat{\operatorname{Conn}}(\hat{P};\hat{\mathfrak{g}}))$ of extensions which has this extended domain is a commutative monoid with pointwise sum. In Subsection \ref{sec_group_ring} we prove that, given a group $\mathbb{G}$, then every $\mathbb{G}$-monoid structure on $\mathbb{R}$ induces for every commutative unital ring $R$ a $R[\mathbb{G}]$-module structure on $\operatorname{Ext}(S^G;\widehat{\operatorname{Conn}}(\hat{P};\hat{\mathfrak{g}}))$. Finally, in Subsection \ref{sec_bundle} we show that by means of varying the extended spaces $\widehat{\operatorname{Conn}}(\hat{P};\hat{\mathfrak{g}})$ we get a $R[\mathbb{G}]$-module bundle with no typical fiber which is continuous in a natural topology if $\hat{P}$ is compact, and we conjecture that it can be regarded as a subbundle of a trivial $R[\mathbb{G}]$-module bundle.

\subsection{Commutative Monoid Structure} \label{sec_monoid}
\quad\;\, Notice that, since an extension is a triple of maps $(\hat{S}^{\hat{G}},C,\delta )$ valued in spaces where addition is well-defined\footnote{Here we are working with a fixed connection $D_0$ in $\hat{P}$, so that $\operatorname{Conn}(\hat{P};\hat{\mathfrak{g}})$ becomes a real vector space.}, we could try to define the sum of two extensions as the pointwise sum of the underlying functions, which are considered in the intersection of the corresponding domains. We clearly have 
$$
(\hat{S}^{\hat{G}}_1+\hat{S}^{\hat{G}}_2)\vert_{C^1_1(\hat{P};\hat{\mathfrak{g}})\cap C^1_2(\hat{P};\hat{\mathfrak{g}})}=S^G\circ (\delta_1+\delta_2)+(C_1+C_2),
$$
but the intersection of two extended domains satisfying (\ref{domain_extension}) need not satisfy (\ref{domain_extension}) nor need to be $\operatorname{Gau}_{\hat{G}}(\hat{P})$-invariant. This leads us to restrict our attention to extensions which are \textit{compatible} meaning that 
$$
\operatorname{Conn}(\hat{P};\hat{\mathfrak{g}}) \subseteq \widehat{\operatorname{Conn}}_1(\hat{P};\hat{\mathfrak{g}})\cap \widehat{\operatorname{Conn}}_2(\hat{P};\hat{\mathfrak{g}}) \subseteq \Omega^{1}_{eq}(\hat{P};\hat{\mathfrak{g}})
$$
and that its intersection is $\operatorname{Gau}_{\hat{G}}(\hat{P})$-invariant. Furthermore, if $\hat{S}^{\hat{G}}_2$ is the some null extension of Example \ref{null_extension}, then $\hat{S}^{\hat{G}}_1+\hat{S}^{\hat{G}}_2$ is the restriction of $\hat{S}^{\hat{G}}_1$ to $C^1_1(\hat{P};\hat{\mathfrak{g}})\cap C^1_2(\hat{P};\hat{\mathfrak{g}})$. Notice that if two extensions have the same extended space, then they are automatically compatible. Let us call a $\operatorname{Gau}_{\hat{G}}(\hat{P})$-invariant set satisfying (\ref{domain_extension}) a \textit{extended domain relative to $\imath:G\hookrightarrow \hat{G}$}. Thus:

\begin{lemma}{\label{monoid_structure}}
For every extended domain $\widehat{\operatorname{Conn}}(\hat{P};\hat{\mathfrak{g}})$ relative to $\imath:G\hookrightarrow \hat{G}$, the collection $\operatorname{Ext}(S^G;\widehat{\operatorname{Conn}}(\hat{P};\hat{\mathfrak{g}}))$ of extensions of $S^G$ which has $\widehat{\operatorname{Conn}}(\hat{P};\hat{\mathfrak{g}})$ as extended space is an abelian monoid with the sum operation above.
\end{lemma}
Let $\hat{S}^{\hat{G}}$ be an extension of a YMT $S^G$ and let $c\in\mathbb{R}$ a real number. We could try to define a new extension by taking the pointwise multiplications $c\cdot\hat{S}^{\hat{G}}$, $c\cdot C$ and $c\cdot \delta$ and the same extended domain and correction subspace. But we do not have 
\begin{equation}\label{scalar_invariant_extension}
(c \cdot \hat{S}^{\hat{G}})\vert_{C^1(\hat{P};\hat{\mathfrak{g}})}=c\cdot (S^G\circ \delta + C),    
\end{equation}
precisely because the action functional of a YMT is not $\mathbb{R}$-linear, but generally quartic. Indeed, assuming without of generality that the pairing is symmetric, 
$$
S^G[c\cdot D]= c^2 \ll dD,dD \gg + 2c^3 \ll dD, \frac{1}{2}D[\wedge]D \gg +c^4 \ll \frac{1}{2}D[\wedge]D, \frac{1}{2}D[\wedge]D \gg,
$$
while 
$$
cS^G[D]= c \ll dD,dD \gg + 2c \ll dD, \frac{1}{2}D[\wedge]D \gg +c \lll \frac{1}{2}D[\wedge]D, \frac{1}{2}D[\wedge]D \gg, 
$$
so that the extensions of a YMT $S^G$ are closed by scalar multiplication iff the action functional $S^D$ is restricted to the subspace $\operatorname{ScConn}(P;\mathfrak{g})\subset \operatorname{Conn}(P;\mathfrak{g})$ of those connections $D$ which are \textit{scalar invariant}, i.e,  such that 
\begin{equation}\label{scalar_connection}
(c^2-c) \ll dD,dD \gg + (2c^3-c) \ll dD, \frac{1}{2}D[\wedge]D \gg +(c^4-c) \ll \frac{1}{2}D[\wedge]D, \frac{1}{2}D[\wedge]D \gg = 0.    
\end{equation}
for every $c\in \mathbb{R}$. But this set is typically empty. Indeed, equation (\ref{scalar_connection}) lead us to look at the roots of the following polynomial over some ring $C(X;\mathbb{R})$ of functions.
\begin{equation}\label{polynomial}
    p(t)=a(x)(t^2-t)+ b(x)(2t^3-t)+c(x)(t^4-t).
\end{equation}

Recall that polynomial rings over fields satisfies the unique polynomial factorization theorem and therefore a degree $n$ univariate polynomial over fields have at most $n$ roots. Therefore, if we look at (\ref{polynomial}) on the field $C_>(X;\mathbb{R})\subset C(X;\mathbb{R})$ of functions which are null or nowhere vanishing we conclude that $p$ has at most $4$ roots. On the other hand, if $b(x)=0=c(x)$ or if $a(x)=0=b(x)$, then the real roots are $t=0,1$. Furthermore, if $a(x)=0=c(x)$, then the roots are $t=0,\pm\sqrt{2}/2$. Let $I(S^G)\subset \mathbb{R}$ be the subset of those $c$ such that (\ref{scalar_connection}) is satisfied for every connection in $P$. The discussion above reveals that under typical hypothesis on $\langle \cdot, \cdot \rangle$, we have $\vert I(S^G)\vert \leq 4$ and that $-1$ is typically not in $I(S^G)$. This implies that the commutative monoid structure of Lemma \ref{monoid_structure} is typically \textit{not} an abelian group. Even so, we can consider its Grothendieck group.

Just to exemplify, let us say that a YMT theory $S^G$ is \emph{semi-nondegenerate} if its pairing $\langle \cdot , \cdot \rangle $ is such that for every connection $D\neq 0$ in $P$ the smooth functions $ \langle dD,dD \rangle $ and $\langle D[\wedge]D, D[\wedge]D \rangle$ are nowhere vanishing (and therefore by continuity strictly positive or strictly negative). This is the case, for instance, if $S^G$ is actually a Yang-Mills theory. A connection $D$ in $P$ is \textit{proper} relatively to the pairing $\langle \cdot, \cdot \rangle$  if $\langle dD,D[\wedge]D\rangle$ is null or nowhere vanishing. Let $\operatorname{PConn}(P;\langle \cdot,\cdot,\rangle)\subset \operatorname{Conn}(P;\mathfrak{g})$ be the space of them. Of course, for semi-nondegenerate YMT theories, flat connections are proper. We say that $S^G$ is \textit{proper} if it is semi-nondegenerate and if each connection is proper relatively to the underlying pairing. E.g, if $G$ is abelian, then every YMT theory $S^G$ which is semi-nondegenerate is also proper.  

\begin{lemma}
Let $S^G$ be a proper YMT theory. Then $\vert I(S^G)\vert \leq 4$ and $I(S^G)=0,1$ if $G$ is abelian.
\end{lemma}
\begin{proof}
Just notice that under the hypothesis of non-degeneracy, since the integral is order-preserving, it follows that the functions $a(D)=\ll dD,dD \gg$ e $c(D)= \ll D[\wedge]D,D[\wedge]D \gg$ are nowhere vanishing. The same happens with $b(D)=\ll dD,D[\wedge]D \gg$ if $D$ is proper. In this case, the polynomial (\ref{polynomial}) with $x=D$ has coefficients in a ring of nowhere vanishing functions, so that it has at most four roots, i.e, $\vert I(S^G)\vert \leq 4$. If $G$ is abelian, then  $D[\wedge]D=0$, so that (\ref{polynomial}) becomes $p(t)=a(D)(t^2-t)$, with roots $t=0,1$. 
\end{proof}
\subsection{$R[\mathbb{G}]$-Module Structure}\label{sec_group_ring}

\quad \;\,In the last section we saw that $\operatorname{Ext}(S^G;\widehat{\operatorname{Conn}}_1(\hat{P};\hat{\mathfrak{g}}))$ is always a commutative monoid with pointwise sum, but almost never closed under the action of $\mathbb{R}$ by pointwise scalar multiplication. It is natural to ask if it can be invariant under other group action. The previous discussion on scalar multiplication give us the hint that the obstruction occurred precisely at the functionals $S^G$ of the YMT theories, which do not commute with the action of $\mathbb{R}$, i.e, we do not have $S^G[c\cdot D]=c\cdot S^G[D]$. Thus, the idea is to consider groups acting in such a way that the property above is satisfied. The ``trick'' here is to remember that we can take pullbacks of action. More precisely, let $\mathbb{G}$ be a group acting on a set $Y$ and let $f:X\rightarrow Y$ a map. We can then define a new map $*^f:\mathbb{G}\times X \rightarrow Y$ by $g*^f x=g*f(x)$, i.e, such that the first diagram below is commutative. In our context, given an action $*$ of $\mathbb{G}$ on $\mathbb{R}$ we get the action $*^{S^G}$ which satisfies the desired property.  
$$
\xymatrix{ &&&&& \ar@/^{1cm}/[rrdrrd]^{*^{(S^G\circ \delta)}} \ar@/^{1.8cm}/[rrdrrd]^{*^{C}} \ar@/^{2.5cm}/[rrdrrd]^{*^{(\hat{S}^{\hat{G}}\circ \jmath)}} \ar@/_{2cm}/[dd]_{S^G\circ \delta} \ar@/_{3.2cm}/[dd]_{C} \ar@/_{4cm}/[dd]_{\hat{S}^{\hat{G}}\circ \jmath} \ar[d]_{id\times \delta} \mathbb{G}\times C^1(\hat{P};\hat{\mathfrak{g}}) \ar@/^{0.1cm}/[rrrrdd]^{(*^{S^G})^{\delta}}   \\ 
\ar[d]_{id\times f} \mathbb{G}\times X \ar@{-->}[rd]^{*^f} &&&&& \ar[d]_{id\times S^G} \mathbb{G}\times \operatorname{Conn}(P;\mathfrak{g}) \ar[rrrrd]^{*^{S^G}}  \\
\mathbb{G}\times \mathbb{R} \ar[r]_-{*} & \mathbb{R} &&&& \mathbb{G}\times \mathbb{R} \ar[rrrr]_-{*} &&&& \mathbb{R}}
$$

Thus, for each extension $\hat{S}^{\hat{G}}$ of $G$ we can form the second diagram above, where in the curved left-hand side arrows we omitted ``$id \times$'' in the labels and $\jmath : C^1(\hat{P};\hat{\mathfrak{g}})\hookrightarrow \widehat{\operatorname{Conn}}(\hat{P};\hat{\mathfrak{g}})$ is the inclusion. Another important fact to note is that pullback of actions is functorial. More precisely, if $f:X\rightarrow Y$ is a map as above and $\delta:X'\rightarrow X$, then $*^{f\circ \delta}=(*^f)^\delta$. Thus, in particular, in the second diagram above we have $*^{(S^G\circ \delta)}=(*^{S^G})^\delta$. Now, notice that the rule $f\mapsto *^f$ preserves the same structures that are preserved by $*$. More precisely, for fixed $X$, let $\operatorname{Map}(X;\mathbb{R})$ be the set of functions $f:X\rightarrow \mathbb{R}$. With pointwise operations it becomes an associative unital real algebra. Let $*^{-}:\operatorname{Map}(X;\mathbb{R})\rightarrow \operatorname{Map}(\mathbb{G}\times X;\mathbb{R})$ be the rule $*^{-}(f)=*^{f}$. As one can quickly check, if $*:\mathbb{G}\times \mathbb{R} \rightarrow \mathbb{R}$ is \textit{additive}, i.e, if it preserves the sum of $\mathbb{R}$, then $*^-$ is too, i.e, $*^{f+f'}=*^f+*^{f'}$, and similarly for the multiplication. In our context, this means that $*^{(S^G\circ\delta+C)}=*^{S^G\circ \delta}+*^{C}$. 

Finally, note that for every fixed $a\in \mathbb{G}$, we have an inclusion $\imath _a:a\times X \hookrightarrow \mathbb{R}\times \mathbb{R}$. Composing with the action $*:\mathbb{G}\times \mathbb{R}\rightarrow \mathbb{R}$ and with the identification $a\times X\simeq X$, for every map $f:X\rightarrow \mathbb{R}$ we have an induced map $a*f:X\rightarrow \mathbb{R}$ given by $(a*f)(x)=(*^f)(a,x)=a*f(x)$. Putting all this together we get:
\begin{theorem}
Let $S^G$ be a YMT theory, $\imath:G\hookrightarrow \hat{G}$ a basic extension of $G$ and $\widehat{\operatorname{Conn}}(\hat{P};\hat{\mathfrak{g}})$ an extended domain. For every group action $*:\mathbb{G}\times \mathbb{R}\rightarrow \mathbb{R}$ we have an induced action 
\begin{equation}\label{induced_group_action}
  \overline{*}:\mathbb{G}\times \operatorname{Ext}(S^G;\widehat{\operatorname{Conn}}(\hat{P};\hat{\mathfrak{g}})) \rightarrow \operatorname{Ext}(S^G;\widehat{\operatorname{Conn}}(\hat{P};\hat{\mathfrak{g}})).  
\end{equation}
Furthermore, if $*$ is additive, then $\overline{*}$ is too with respect to the commutative monoid structure of Lemma \ref{monoid_structure}.
\end{theorem}
\begin{proof}
After the previous discussion the proof is basically done. Indeed, if $\mathcal{S}=(\hat{S}^{\hat{G}},C,\delta)$ is an extension of $S^G$ with extended domain $\widehat{\operatorname{Conn}}(\hat{P};\hat{\mathfrak{g}})$, for each $a\in \mathbb{G}$ define 
$
a\overline{*}\mathcal{S}=(a\overline{*}\hat{S}^{\hat{G}},a\overline{*}C,a\overline{*}\delta)
$, where 
$$
(a\overline{*}\hat{S}^{\hat{G}})(\alpha)=\begin{cases} [a*(\hat{S}^{\hat{G}}\circ \jmath)](\beta) & \text{if}\;\alpha=\jmath (\beta),\;\text{with}\; \beta \in C^1(\hat{P};\hat{\mathfrak{g}}) \\ \hat{S}^{\hat{G}}(\alpha) & \text{otherwise}, \end{cases}
$$
while $a\overline{*}C=a*C$ and $[a\overline{*}\delta](\beta)=[a*^{S^G}\delta](\beta)=a*[S^G(\delta(\beta))]$. One can check that, if defined in this way, $a\overline{*}\mathcal{S}$ is another extension of $S^G$ which by construction has the same extended domain. Furthermore, by the above expressions we see that $e\overline{*}\mathcal{S}=\mathcal{S}$, where $e$ is the neutral element of $\mathbb{G}$, and that $a\overline{*}(b\overline{*}\mathcal{S})=(a*b)\overline{*}\mathcal{S}$, so that $\overline{*}$ really defines the desired group action (\ref{induced_group_action}). Finally, since the commutative monoid structure of Lemma \ref{monoid_structure} is given by pointwise sum, the previous remark that $*$ is additive implies $*^{f+f'}=*^f+*^{f'}$ which shows that $\overline{*}$ is additive if $*$ is, finishing the proof.
\end{proof}
Now, let $X$ be some additive commutative monoid endowed with an additive group action by $\mathbb{G}$, i.e, let $X$ be a $\mathbb{G}$-monoid. Recall that for every $\mathbb{G}$ the category of $\mathbb{G}$-monoids can be identified with the category of $\mathbb{Z}[\mathbb{G}]$-monoids, where $\mathbb{Z}[\mathbb{G}]$ is the group ring of $\mathbb{G}$, which is unital. Furthermore, notice that every $R$-monoid, where $R$ is a unital ring with unit $e$, can be turned a $R$-group, i.e, an $R$-module. Indeed, for every $x\in X$ define its additive inverse by $-x:=(-e)*x$, where $-e$ is the additive inverse of $e$ in $R$\footnote{We clearly have $0=0*x=[e+(-e)]*x=(e*x)+[(-e)*x]=x+(-x)=(-x)+x$.}. Thus, 

\begin{corollary}
For every ring $R$ and every additive group action $*:\mathbb{G}\times \mathbb{R}\rightarrow \mathbb{R}$, i.e, for every $\mathbb{G}$-monoid structure in $\mathbb{R}$, we have an induced $R[\mathbb{G}]$-module structure in  $\operatorname{Ext}(S^G;\widehat{\operatorname{Conn}}(\hat{P};\hat{\mathfrak{g}}))$.
\end{corollary}
\begin{proof}
By the above discussion, for every $*$ we have a $\mathbb{Z}[\mathbb{G}]$-module structure. Then recall that $\mathbb{Z}[X]\otimes_{\mathbb{Z}}R\simeq R[X]$.
\end{proof}

\begin{remark}
If an action $*:\mathbb{G}\times \mathbb{R}\rightarrow \mathbb{R}$ is not only additive, but also multiplicative, then it corresponds to a 1-dimensional real irreductive representation of $\mathbb{G}$. Thus, in this setting, or the action is trivial (which is not desired) or $\mathbb{G}$ is not semi-simple.
\end{remark}

\subsection{Towards a Bundle Structure}\label{sec_bundle}

\quad\;\, Fixed a YMT theory $S^G$ and a basic extension $\imath : G\hookrightarrow \hat{G}$ for $G$, let $\operatorname{ED}(S^G;\hat{G})$ be the set of all extended domains $\widehat{\operatorname{Conn}}(\hat{P};\hat{\mathfrak{g}})$. We have a projection $\pi:\operatorname{Ext}(S^G;\hat{G})\rightarrow \operatorname{ED}(S^G;\hat{G})$ assigning to each extension of $S^G$ its underlying extended domain. The fibers of this projection are precisely the sets $\operatorname{Ext}(S^G;\widehat{\operatorname{Conn}}(\hat{P};\hat{\mathfrak{g}}))$. But in the last subsection we proved that an additive $\mathbb{G}$-action in $\mathbb{R}$ induces a $R[\mathbb{G}]$-module structure in each of these fibers. Thus, it is very natural to ask whether $\pi:\operatorname{Ext}(S^G;\hat{G})\rightarrow \operatorname{ED}(S^G;\hat{G})$ is actually a $R[\mathbb{G}]$-module bundle.

First of all, notice that by the definition of extended domain, one can regard $\operatorname{ED}(S^G;\hat{G})$ as the subset of the power set $\mathcal{P}(\Omega^1(\hat{P};\hat{\mathfrak{g}}))$ consisting of those subsets $X\subset \Omega^1(\hat{P};\hat{\mathfrak{g}})$ satisfying (\ref{domain_extension}). Since $\Omega^1(\hat{P};\hat{\mathfrak{g}})\simeq \Gamma(T^{*}\hat{P}\otimes (\hat{P}\times \hat{\mathfrak{g}}))$, if $\hat{P}$ is compact, then $\Omega^1(\hat{P};\hat{\mathfrak{g}})$ becomes a Fr\'echet space and therefore a complete metric space \cite{michor_global_analysis}. In particular, it can be regarded as a complete uniform space, which means that its power set $\mathcal{P}(\Omega^1(\hat{P};\hat{\mathfrak{g}}))$ has an induced uniform structure and, in particular, a natural (and typically non-Hausdorff) topology \cite{uniform_spaces}. One can then take in $\operatorname{ED}(S^G;\hat{G}$) the subspace topology. Let $\Pi:\operatorname{Ext}(S^G;\hat{G})\rightarrow \mathcal{P}(\Omega^1(\hat{P};\hat{\mathfrak{g}}))$ be the composition of $\pi$ with the inclusion $\operatorname{ED}(S^G;\hat{G}) \hookrightarrow \mathcal{P}(\Omega^1(\hat{P};\hat{\mathfrak{g}}))$ and give $\operatorname{Ext}(S^G;\hat{G})$ the initial topology defined by $\Pi$. Thus, the closed sets of $\operatorname{Ext}(S^G;\hat{G})$ are preimages of closed sets on $\mathcal{P}(\Omega^1(\hat{P};\hat{\mathfrak{g}}))$. Furthermore, since $\pi$ is surjective, it follows that it is continuous in the initial topology of $\Pi$. This topology is better than the initial topology of $\pi$ because it is completely determined on the subspace of closed sets of $\mathcal{P}(\Omega^1(\hat{P};\hat{\mathfrak{g}}))$, which constitute a more well-behaved environment. E.g, the induced uniform structure arises from a metric, which is complete since  $\Omega^1(\hat{P};\hat{\mathfrak{g}})$ a Fr\'echet space \cite{michor_global_analysis,uniform_spaces}. 

Thus, fixed a $\mathbb{G}$-monoid structure in $\mathbb{R}$ we get, for every commutative unital ring $R$, a continuous projection $\pi_R:\operatorname{Ext}(S^G;\hat{G})\rightarrow \operatorname{ED}(S^G;\hat{G})$ whose fibers are $R[\mathbb{G}]$-modules. We should not expect that this projection is locally trivial, since for different extended domains we have fibers with different sizes, so that a priori there is no typical fiber. On the other hand, we speculate that it can be regarded as a subbundle of a trivial bundle. More precisely, notice that $\operatorname{ED}(S^G;\hat{G})$ is a directed set with the inclusion relation. Furthermore, given two extended domains such that  $\imath_{1,2}:\widehat{\operatorname{Conn}}_1(\hat{P};\hat{\mathfrak{g}})\hookrightarrow \widehat{\operatorname{Conn}}_2(\hat{P};\hat{\mathfrak{g}})$ we have a corresponding map   \begin{equation}\label{fiberwise_map}
    \imath_{1,2}^{*}:\operatorname{Ext}(S^G;\widehat{\operatorname{Conn}}_2(\hat{P};\hat{\mathfrak{g}}))\rightarrow \operatorname{Ext}(S^G;\widehat{\operatorname{Conn}}_1(\hat{P};\hat{\mathfrak{g}}))
\end{equation}
given by restriction, i.e, such that for every $\mathcal{S}_2=(\hat{S}_2^{\hat{G}},C_2,\delta_2)$ we have $\imath _{1,2}^{*}\mathcal{S}_2=(\hat{S}_2^{\hat{G}}\circ \imath_{1,2},C_2\circ \imath_{1,2},\delta_2\circ \imath_{1,2})$, where the compositions are in the proper domains. One can check that this map actually preserves the $R[\mathbb{G}]$-module structures of last section. Thus, we have an inverse limit in the category of $R[\mathbb{G}]$-modules, whose limit is $\operatorname{Ext}(S^G;\Omega^1_{eq}(\hat{P};\hat{\mathfrak{g}}))$, since the direct set $\operatorname{ED}(S^G;\hat{G})$ has $\Omega^1_{eq}(\hat{P};\hat{\mathfrak{g}})$ as its greatest element. 

Now, consider the projection \begin{equation}\label{projection_trivial_bundle}
    \operatorname{pr}_1:\operatorname{ED}(S^G;\hat{G})\times \operatorname{Ext}(S^G;\Omega^1_{eq}(\hat{P};\hat{\mathfrak{g}}))\rightarrow \operatorname{ED}(S^G;\hat{G}).
\end{equation}
Let $\operatorname{PR}_1$ be the composition of $\operatorname{pr}_1$ with the inclusion $\operatorname{ED}(S^G;\hat{G})\hookrightarrow \mathcal{P}(\Omega^1(\hat{P};\hat{\mathfrak{g}}))$ and consider its initial topology. Thus, (\ref{projection_trivial_bundle}) becomes continuous and defines a continuous trivial $R[\mathbb{G}]$-module bundle. We have a map $\imath^*$ such that $\imath^*(\widehat{\operatorname{Conn}}(\hat{P};\hat{\mathfrak{g}}),\mathcal{S})=(\widehat{\operatorname{Conn}}(\hat{P};\hat{\mathfrak{g}}),\imath ^{*}\mathcal{S})$, i.e, which is fiberwise of the form (\ref{fiberwise_map}), clearly making commutative the diagram below.
\begin{equation}\label{diagram_bundle_1}
   \xymatrix{\ar[rd]_{\operatorname{pr}_1}\operatorname{ED}(S^G;\hat{G})\times \operatorname{Ext}(S^G;\Omega^1_{eq}(\hat{P};\hat{\mathfrak{g}})) \ar[r]^-{\imath ^{*}} & \operatorname{Ext}(S^G;\hat{G}) \ar[d]^{\pi_R} \\
& \operatorname{ED}(S^G;\hat{G})} 
\end{equation}

Notice that, if $\imath^*$ is continuous, it then defines a morphism of $R[\mathbb{G}]$-module bundles. Furthermore, suppose that we found $R$ and $\mathbb{G}$ such that $\operatorname{Ext}(S^G;\Omega^1_{eq}(\hat{P};\hat{\mathfrak{g}}))$ is an injective module. Thus, $\imath^{*}$ becomes fiberwise an epimorphism of $R[\mathbb{G}]$-modules. If, in addition each $\operatorname{Ext}(S^G;\widehat{\operatorname{Conn}}(\hat{P};\hat{\mathfrak{g}}))$ is a projective $R[\mathbb{G}]$-module (except maybe for $\widehat{\operatorname{Conn}}(\hat{P};\hat{\mathfrak{g}})= \Omega^1_{eq}(\hat{P};\hat{\mathfrak{g}})$), then $\imath^*$ is fiberwise an split epimorphism and we can define a section $s$ for $\imath^{*}$ in such a way that diagram (\ref{diagram_bundle_1}) extends to the following commutative diagram. Furthermore, if the fiberwise splits of $\imath ^*$ are such that $s$ is continuous, then $s$ embeds the bundle $\pi_R$ of extensions of $S^G$ as a $R[\mathbb{G}]$-subbundle of the trivial $R[\mathbb{G}]$-bundle (\ref{projection_trivial_bundle}). 
\begin{equation}\label{diagram_bundle_2}
   \xymatrix{\ar[rd]_{\operatorname{pr}_1}\operatorname{ED}(S^G;\hat{G})\times \operatorname{Ext}(S^G;\Omega^1_{eq}(\hat{P};\hat{\mathfrak{g}})) \ar[r]^-{\imath ^{*}} & \operatorname{Ext}(S^G;\hat{G}) \ar[d]^{\pi_R} \ar[r]^-{s} & \ar[ld]^{\operatorname{pr}_1}\operatorname{ED}(S^G;\hat{G})\times \operatorname{Ext}(S^G;\Omega^1_{eq}(\hat{P};\hat{\mathfrak{g}})) \\
& \operatorname{ED}(S^G;\hat{G})} 
\end{equation}

Thus, after all the previous discussion we are lead to speculate the following:
\begin{conjecture}
For every YMT theory $S^G$, every basic extension $\imath:G\hookrightarrow \hat{G}$ of $G$ such that $\hat{P}$ is compact, every ring $R$ and every group $\mathbb{G}$, the corresponding map $\imath^{*}$ defined above is continuous in the described topologies.
\end{conjecture}
\begin{conjecture}
Given a YMT theory $S^G$ and a basic extension  $\imath:G\hookrightarrow \hat{G}$ of $G$ such that $\hat{P}$ is compact, there exists a ring $R$ and a group $\mathbb{G}$ such that $\operatorname{Ext}(S^G;\Omega^1_{eq}(\hat{P};\hat{\mathfrak{g}}))$ is injective as $R[\mathbb{G}]$-module and such that $\operatorname{Ext}(S^G;\widehat{\operatorname{Conn}}(\hat{P};\hat{\mathfrak{g}}))$  (except maybe for $\widehat{\operatorname{Conn}}(\hat{P};\hat{\mathfrak{g}})=\Omega^1_{eq}(\hat{P};\hat{\mathfrak{g}})$) is a projective as $R[\mathbb{G}]$-module. Furthermore, the induced splits can be choosen such that the corresponding map $s$ of diagram (\ref{diagram_bundle_2}) is continuous in the described topologies.  
\end{conjecture}

\section{Category of Extensions}\label{sec_category_extensions}

\quad\;\,In previous sections we studied the space of extensions. Here we will study how the elements of this space interact with each other. More precisely, we will build the category of extensions of a given YMT theory.

Let $S^G$ be a YMT theory and let $\imath:G \hookrightarrow \hat{G}$ be a basic extension of $G$, which from now on will be fixed. Let $\hat{S}^{\hat{G}}_i\in \operatorname{Ext}(S^G;\hat{G})$, with $i=1,2$ be extensions of $S^G$. A \textit{morphism} between them, denoted by $F:\hat{S}^{\hat{G}}_1 \rightarrow \hat{S}^{\hat{G}}_2$ is given by a pair $F=(f,g)$, where $f:\widehat{\operatorname{Conn}}_{1}(\hat{P};\hat{\mathfrak{g}})\rightarrow\widehat{\operatorname{Conn}}_{2}(\hat{P};\hat{\mathfrak{g}})$ is a $\operatorname{Gau}_{\hat{G}}(\hat{P})$-equivariant map
and $g:C_{1}^{1}(\hat{P};\hat{\mathfrak{g}})\rightarrow C_{2}^{1}(\hat{P};\hat{\mathfrak{g}})$ is a function
making commutative the first diagram below. Compositions and identities are defined componentwise, i.e, $(f,g)\circ (f',g')=(f'\circ f, g\circ g)$. It is straightforward to check that this data really defines a category $\mathbf{Ext}(S^G;\hat{G})$.
\begin{equation}{\label{diagram_equivalence_extensions}
\xymatrix{   && \ar@/_/[lld]_-{\delta _1} \ar[dd]^g C^1_1(\hat{P};\hat{\mathfrak{g}}) \ar[ld]_{C_1} \ar@{^(->}[rr] && \widehat{\operatorname{Conn}}_1(\hat{P};\hat{\mathfrak{g}}) \ar[ld]_-{\hat{S}^{\hat{G}}_1} \ar[dd]^f \\
\operatorname{Conn}(P;\mathfrak{g}) & \mathbb{R} & & \mathbb{R} \\
 && \ar@/^/[llu]^-{\delta _2} \ar[lu]^{C_2} C^1_2(\hat{P};\hat{\mathfrak{g}}) \ar@{^(->}[rr] &&   \widehat{\operatorname{Conn}}_2(\hat{P};\hat{\mathfrak{g}}) \ar[lu]^-{\hat{S}^{\hat{G}}_2}}}
\end{equation}

Given a group $G$, let $\mathbf{Set}_G$ denote the category of $G$-sets, i.e, whose objects are sets $X$ with an action $G\times X\rightarrow X$ and whose morphisms are $G$-equivariant maps. Furthermore, if $\mathbf{C}$ is a category and $A\in\mathbf{C}$ is an object, let $\mathbf{C}/A$ denote the corresponding slice category, whose objets are morphisms $f:X \rightarrow A$ to $A$ and whose morphisms $h:f \Rightarrow g$ are commutative triangles, as below, i.e, are morphisms $h:X\rightarrow Y$ in $\mathbf{C}$ such that $f=h\circ g$. We also recall that a subcategory $\mathbf{C}_0\subset \mathbf{C}$ is \textit{conservative} if the inclusion functor $\imath:\mathbf{C}_0\rightarrow \mathbf{C}$ is conservative, i.e, if reflects isomorphisms, meaning that if a map $f:X\rightarrow Y$ is such that $\imath(f)$ is an isomorphism in $\mathbf{C}$, then $f$ is an isomorphism in $\mathbf{C}_0$.
$$
\xymatrix{& \ar[ld]_f X \ar[dd]^{h} \\
A \\
& Y \ar[lu]^g }
$$
\begin{theorem}\label{prop_prop_conservative}
For every given YMT $S^G$ and basic extension $\imath:G\hookrightarrow \hat{G}$, the corresponding category of extensions $\mathbf{Ext}(S^G;\hat{G})$ is a conservative subcategory of 
\begin{equation}\label{prop_conservative_subcategory}
 (\mathbf{Set}_{\operatorname{Gau}_{\hat{G}}(\hat{P})}/\mathbb{R})\times (\mathbf{Set}/\mathbb{R})\times (\mathbf{Set}/\operatorname{Conn}(P;\mathfrak{g})).   
\end{equation}

\end{theorem}
\begin{proof}
First of all, notice that an object of $\mathbf{Ext}(S^G;\hat{G})$ can be regarded as a triple $(\hat{S}^{\hat{G}},C,\delta)$ in (\ref{prop_conservative_subcategory}) such that $s(C)=s(\delta)$, $s(\delta)\subset s(\hat{S}^{\hat{G}})$ and $\hat{S}^{\hat{G}}\vert_{s(C)}=S^G\circ \delta +C$, where $s(f)$ denoted the domain of a map and 
$$
\operatorname{Conn}(\hat{P};\hat{\mathfrak{g}})\subset s(\hat{S}^{\hat{G}}) \subset \Omega ^1_{eq}(\hat{P};\hat{\mathfrak{g}}).
$$ 
Furthermore, looking at diagram (\ref{diagram_equivalence_extensions}) a morphism $(f,g)$ in $\mathbf{Ext}(S^G;\hat{G})$ can be regarded as a morphism $(f,g,h)$ in (\ref{prop_conservative_subcategory}) such that $h=g$ and $f\vert _{s(C)}=g$. Since compositions and identities in $\mathbf{Ext}(S^G;\hat{G})$ are componentwise, it follows that $\mathbf{Ext}(S^G;\hat{G})$ is a subcategory of (\ref{prop_conservative_subcategory}). It is straighforward to check that the inclusion is faithful. Therefore, it reflects monomorphisms and epimorphisms. But epimorphisms  (resp. monomorphisms) on the category of sets and $G$-sets are surjective (resp. monomorphism) maps \cite{category_G_sets}. Since in slice category limits and colimits are computed objectwise, we see that $(f,g,h)$ is an epimorphism (resp. monomorphism) iff $f$, $g$ and $h$ are surjections (resp. injections). Thus, a morphism $(f,g)$ in $\mathbf{Ext}(S^G;\hat{G})$ is an epimorphism (resp. monomorphism) iff $f$ and $g$ are surjections (resp. injections). In particular, $(f,g)$ is a bimorphism iff both $f$ and $g$ are bijections. But it is straighforward to check that, in this case, $(f,g)$ is an isomorphism. Thus, $\mathbf{Ext}(S^G;\hat{G})$ is balanced. The result then follows from the fact that very faithful functor defined on a balanced category is conservative \cite{riehl_category}.  
\end{proof}
\begin{corollary}
Suppose that a diagram $J$ in $\mathbf{Ext}(S^G;\hat{G})$ admits limit (resp. colimit) and is such that for every limiting cone (resp. cocone) $L$ the image $\imath(L)$ is a limiting cone (resp. cocone) for $\imath \circ J$, where $\imath$ is the inclusion functor to \emph{(\ref{prop_conservative_subcategory})}. In this case, the reciprocal holds, i.e, if $L$ is such that $\imath(L)$ is a limiting cone (resp. cocone), then $L$ must be a limiting cone (resp. cocone). 
\end{corollary}
\begin{proof}
It follows from the general fact that conservative functor $F:\mathbf{C}\rightarrow \mathbf{D}$ reflects every limit and colimit that exists in $\mathbf{C}$ and which are preserved by $F$.
\end{proof}

There are, however, some limits and colimits which exists in $\mathbf{Ext}(S^G;\hat{G})$ but which are not preserved (and therefore not reflected) by the inclusion functor.

\begin{example}
\emph{As it is straightforward to check, the null extension of Example \ref{null_extension} is a terminal object in the category of extensions $\mathbf{Ext}(S^G;\hat{G})$. We assert that it is not preserved by $\imath$ if $S^G$ is not null. Indeed, a terminal object in a slice category $\mathbf{C}/A$ is the identity $id_A$, so that a terminal object in (\ref{prop_conservative_subcategory}) is the triple $id=(id_{\mathbb{R}},id_{\mathbb{R}},id_{\operatorname{Conn}(P;\mathfrak{g})})$. This means that if the inclusion functor preserves terminal objects, then  $id$ must belong to $\mathbf{Ext}(S^G;\hat{G})$, i.e, there exists an extension of $S^G$ such that $\hat{S}^{\hat{G}}=id_\mathbb{R}$, $C=id_\mathbb{R}$ and $\delta=id_{\operatorname{Conn}(P;\mathfrak{g})}$. But the condition (\ref{decomposition_extension}) would implies $S^G=0$, which is a contradiction.}
\end{example}

\begin{remark}
In the case of equivariant extensions (see Remark \ref{remar_linear_equivariant_extensions}), Theorem \ref{prop_prop_conservative} can be improved by replacing $\mathbf{Set}/\mathbb{R}$ with $\mathbf{Set}_{\operatorname{Gau}_{\hat{G}}(\hat{P})}/\mathbb{R}$ and $\mathbf{Set}/\operatorname{Conn}(P;\mathfrak{g})$ with $\mathbf{Set}_{\operatorname{Gau}_{\hat{G}}}/\operatorname{Conn}(P;\mathfrak{g})$, where the $\operatorname{Gau}_{\hat{G}}(\hat{P})$-action in $\operatorname{Conn}(P;\mathfrak{g})$ is that induced by the homorphism $\xi:\operatorname{Gau}_G(P)\rightarrow \operatorname{Gau}_{\hat{G}}(\hat{P})$ described in Section \ref{sec_extensions}.
\end{remark}
\begin{remark}[speculation]
In Subsection \ref{sec_bundle} we proved that $\operatorname{Ext}(S^G;\hat{G})$ is an $R[\mathbb{G}$-module bundle for every ring $R$ and every additive group action $\mathbb{G}$ in $\mathbb{R}$ and we conjectured that it can be regarded as continuous subundle of a trivial bundle. In the present section, on the other hand, we showed that $\operatorname{Ext}(S^G;\hat{G})$ is the collection of objects of a category $\mathbf{Ext}(S^G;\hat{G})$. Thus, it is natural to ask whether the collection of morphisms $\operatorname{Mor}(S^G;\hat{G})$ of $\mathbf{Ext}(S^G;\hat{G})$ can also be endowed with a $R[\mathbb{G}]$-module bundle structure such that the source, target and composition maps are bundle morphisms. In other words, it is natural to ask if $\mathbf{Ext}(S^G;\hat{G})$ cannot be internalized in the category $\mathbf{Bun}_{R[\mathbb{G}]}$ of $R[\mathbb{G}]$-module bundles. We do not have a complete answer and a careful analysis could motivate future research.
\end{remark}

\section*{Acknowledgements} The authors would like to thank Helvecio Fargnoli Filho for reading a preliminary version of the text. Y. X. Martins was supported by CAPES (grant number 88887.187703/2018-00) and L. F. A. Campos was supported by CAPES (grant number 88887.337738/2019-00).

\bibliographystyle{abbrv}
\bibliography{references}
\end{document}